\documentclass[lettersize,journal]{IEEEtran}
\usepackage{cite}
\usepackage{amsmath,amssymb,amsfonts}
\usepackage{amsthm}
\usepackage{graphicx}
\usepackage{textcomp}
\usepackage{graphicx}   
\usepackage{mathtools} % Bonus
\usepackage{diagbox}
\usepackage{hyperref}
\usepackage{float}
\usepackage{amsmath}
\usepackage{algpseudocode}
\usepackage{caption}
\usepackage{algorithm}  
\usepackage{setspace}
\usepackage{bm}
\usepackage{enumerate}
\usepackage{algorithm}
\usepackage{algpseudocode}

\usepackage{multirow}
\usepackage{subfigure}
\usepackage{bm} 
\usepackage{makecell}

\usepackage[utf8]{inputenc}
\usepackage{multirow}
\usepackage{algpseudocode}
\usepackage{stfloats}
\usepackage{mathrsfs}
\usepackage{xcolor}

\newtheorem{definition}{Definition}
\newtheorem{lemma}{Lemma}
\newtheorem{remark}{Remark}
\allowdisplaybreaks[4]

\begin{document}

\title{User Assignment and Resource Allocation for Hierarchical Federated Learning over Wireless Networks
\thanks{T. Zhang, K. Lam, and J. Zhao are with the Strategic Research Centre for Research in Privacy-Preserving Technologies and Systems, and the School of Computer Science and Engineering at Nanyang Technological University, Singapore. (Emails: tinghao001@e.ntu.edu.sg; kwokyan.lam@ntu.edu.sg; junzhao@ntu.edu.sg).

% F. Li is with the Strategic Research Centre for Research in Privacy-Preserving Technologies and Systems at Nanyang Technological University, Singapore. (Email: fengli2002@yeah.net).

% Norziana Jamil is with College of Computing \& +ormatics, University Tenaga Nasional, Putrajaya Campus,
% Selangor, Malaysia (Email: Norziana@uniten.edu.my).
}
}

\author{Tinghao Zhang, Kwok-Yan Lam, \emph{Senior Member}, IEEE, Jun Zhao\vspace{-20pt}}

\maketitle
\begin{abstract}
The large population of wireless users is a key driver of data-crowdsourced Machine Learning (ML). However, data privacy remains a significant concern. Federated Learning (FL) encourages data sharing in ML without requiring data to leave users’ devices but imposes heavy computation and communications overheads on mobile devices. Hierarchical FL (HFL) alleviates this problem by performing partial model aggregation at edge servers. HFL can effectively reduce energy consumption and latency through effective resource allocation and appropriate user assignment. Nevertheless, resource allocation in HFL involves optimizing multiple variables, and the objective function should consider both energy consumption and latency, making the development of resource allocation algorithms very complicated. Moreover, it is challenging to perform user assignment, which is a combinatorial optimization problem in a large search space. This article proposes a spectrum resource optimization algorithm (SROA) and a two-stage iterative algorithm (TSIA) for HFL. Given an arbitrary user assignment pattern, SROA optimizes CPU frequency, transmit power, and bandwidth to minimize system cost. TSIA aims to find a user assignment pattern that considerably reduces the total system cost. Experimental results demonstrate the superiority of the proposed HFL framework over existing studies in energy and latency reduction.
\end{abstract}

\maketitle

\section{Introduction}

% \the\abovedisplayskip

% \the\belowdisplayskip

% \the\abovedisplayshortskip

% \the\belowdisplayshortskip 

Machine learning (ML) has been widely used in a variety of areas such as image classification, natural language processing, and autonomous control~\cite{LeCun15}. ML usually requires a massive amount of training data to achieve satisfactory model performance. In recent years, data crowdsourcing has emerged as an important paradigm for acquiring training data needed by ML~\cite{Ma20,Yang22,Pandey20}. In particular, mobile application service providers have been deploying innovative reward schemes to incentivize mobile users to participate in data crowdsourcing. As part of the efforts to develop smart cities and promote the digital economy, many government agencies are also adopting data crowdsourcing to acquire data to fuel ML. Given the pervasive adoption of smart mobile devices, the large population of wireless mobile users plays a key role in data-crowdsourced ML. In the most straightforward manner, all participating mobile users transmit their crowdsourced data to some cloud server for centralized training, which is bottlenecked by limited bandwidth resources and privacy concerns~\cite{Liu22}. Furthermore, there are major concerns of the users that need to be addressed to gain people’s trust and remove the hurdles in supporting data crowdsourcing~\cite{Burg18}. These issues include cybersecurity of the central server, data privacy of the users and resource consumption on their mobile devices. In particular, battery and bandwidth consumption are of particular concern to most mobile users~\cite{Zhang22, Yang21, Li22}.

To address the problems of centralized training, a decentralized ML framework called Federated Learning (FL) is widely adopted as an approach for encouraging data sharing in ML without requiring data to leave the devices of the participants~\cite{McMahan17}. In FL, mobile users train the local model on their datasets, and the model weights are transmitted to a remote server for model aggregation. Compared with centralized training, FL provides better protection against privacy leakage as the training data are kept on mobile users~\cite{Ziyao22, Zhao21,yang2022lead}. While enhancing data privacy, however, FL requires a non-trivial amount of computation to be performed on the mobile device, hence inevitably leading to high energy consumption and depleting device batteries rapidly. Besides, uploading local models to the aggregation server consumes valuable wireless bandwidth. In addition, due to the bandwidth issues, FL needs to be designed carefully to scale up the size of data crowdsourcing, which will directly impact the model performance. FL relying on a single server for model aggregation tends to cause network congestion~\cite{Bonawitz19}.

Recently, the deployment architecture of FL was extended to address multi-server scenarios. The novel FL framework, namely Hierarchical Federated Learning (HFL), consists of several edge servers and one cloud server. Mobile users in HFL train the models locally and then send their model weights to the corresponding edge servers for edge aggregation. After several edge iterations, the edge servers transmit the edge models to the cloud for global aggregation. HFL has great potential in alleviating communication overheads. More importantly, the cybersecurity protection of HFL can be implemented more flexibly, thus reducing the impact on accuracy and computations~\cite{Wainakh20}. 

Despite the advantages, some issues of HFL have not been fully investigated yet. For example, the communication and computation resources should be properly allocated to mobile users. As mobile users will inevitably consume higher energy if pursuing a lower time delay, it is crucial for resource allocation to balance the tradeoff between the time delay and energy consumption. In addition, resource allocation is anticipated to jointly optimize both local resources (e.g., CPU frequency and transmit power of mobile users) and system resources (e.g., network bandwidth), as these resources have a significant impact on the energy consumption and latency of the HFL system. As a result, designing an efficient resource allocation algorithm for HFL is difficult due to its complex objective and the presence of multiple variables to be jointly optimized. Besides performing resource allocation, each mobile user will be allocated to an edge server before training. Without a feasible user assignment algorithm, some edge servers may experience prolonged latency during edge updates caused by straggling mobile users. As the cloud server conducts global aggregation only when receiving the model weights from all the edge servers, the straggling mobile users will finally deteriorate the overall latency of HFL. Moreover, if the straggling mobile users increase computation frequency to reduce the computational latency, the energy consumption to perform local update will notably rise. Therefore, user assignment plays a crucial role in reducing the energy consumption and latency of HFL. Whereas, it is hard to directly solve the user assignment problem since it is a combinatorial optimization problem with a large search space. Thus, how to carry out user assignment is another challenge of deploying HFL. 

This paper proposes a spectrum resource optimization algorithm (SROA) to handle the resource allocation problem and a two-stage iterative algorithm (TSIA) to deal with the user assignment problem in HFL. The main contributions of this paper are summarized as follows:

\begin{enumerate}
    \item We formulate a joint communication and computation optimization problem to minimize the total system cost, which is represented as the weighted sum of the time delay and energy consumption for training the entire HFL algorithm. To address this problem, we propose a novel HFL framework that contains two key modules: SROA-based resource allocation and TSIA-based user assignment. 
    
    \item We propose a spectrum resource optimization algorithm (SROA) to cope with the resource allocation problem in HFL. Given a user assignment pattern, SROA aims to minimize the total system cost by optimizing the CPU frequency, average transmit power, and the bandwidth allocated to edge servers and mobile users. To this end, SROA adopts the divide-and-conquer paradigm and utilizes binary search methods to obtain the locally optimal solution to the resource allocation problem.
    
    \item We propose a two-stage iterative algorithm (TSIA) for the user assignment problem. TSIA adjusts the user assignment pattern by iteratively transforming users from resource-intensive edge servers to more economical servers, thereby achieving a more balanced workload distribution. Among the patterns TSIA has explored, the user assignment pattern with the lowest total system cost is chosen as the solution to the user assignment problem.
     
    \item We conduct extensive numerical experiments to evaluate the proposed methods. Given a user assignment pattern, SROA achieves a lower total system cost than the existing resource allocation methods. TSIA outperforms the state-of-the-art user assignment methods in reducing the total system cost and convergence speed. Moreover, the proposed HFL framework achieves a lower total system cost than traditional FL.
\end{enumerate}
The rest of the paper is organized as follows. The related work is listed in Section~\ref{section2}. The system model and problem formulation are formulated in Section~\ref{section3}. The SROA algorithm is described in Section~\ref{resource_allo}. The TSIA algorithm is explained in Section~\ref{edge_allo}. Experimental results are analyzed in Section~\ref{section6} followed by the conclusion in Section~\ref{section7}.

\section{Related Work}\label{section2}

Federated Learning has become a mainstream topic in the field of machine learning. Some existing researches focus on defense algorithms such as differential privacy~\cite{Wei20,Truex20,Seif20}, secure multi-party computation~\cite{Wenhao21,Kanagavelu20,Byrd20}, and homomorphic encryption~\cite{Zhang21,Haokun21,Madi21}, to improve user security in FL. However, such defences will either drop the accuracy or incur additional computation overheads~\cite{Hongkyu21,Liu21,Huang21}. In contrast, HFL allows the defences to be deployed more flexibly. For example, it is unnecessary to implement secure aggregation for the upper layers of HFL as the cloud only processes the edge models and has no access to the users' update~\cite{Wainakh20}. 

Some researchers investigate resource allocation problems to alleviate the communication and computation overheads in FL. For example, \cite{Xu21} proposes a joint user scheduling and bandwidth allocation method to pursue better learning performance of FL. \cite{Zeng20Sky} designs a joint power allocation and scheduling algorithm to optimize the convergence speed of FL. \cite{Chen21} minimizes the time delay for training the entire FL algorithm by introducing a joint resource allocation and user scheduling scheme. \cite{Yang20} balances the tradeoff between local computation delay and wireless transmission delay to minimize the time delay of FL. \cite{Shi21} aims to minimize the time delay by optimizing bandwidth allocation. Nevertheless, \cite{Xu21, Zeng20Sky, Chen21, Yang20, Shi21} do not optimize the energy consumption in FL. \cite{Zeng20} conducts bandwidth allocation based on the channel conditions and computation capacities. However, the time delay is not optimized in~\cite{Zeng20}. \cite{Dinh20} proposes Federated Learning over wireless networks problem (FEDL) to balance the energy consumption and time delay of training FL. In traditional FL, however, the mobile users upload the model weights to a single remote server in these works, which still faces high communication overheads. HFL can mitigate the network congestion by partially transferring the aggregation processes to multiple edge servers~\cite{Luo20}. 

In terms of HFL, \cite{Liu20} proposes a basic HFL framework and theoretically analyzes the convergence rate of HFL. \cite{Wainakh20} analyzes the potential of HFL for addressing privacy issues in traditional FL. \cite{Jinliang20} presents an HFL framework that adopts a local-area network and a wide-area network to accelerate the training process and alleviate network traffic. However, resource allocation and user assignment are not discussed in the aforementioned studies. \cite{Mhaisen21} proposes a user assignment algorithm for HFL to speed up the convergence speed of HFL on the non-IID datasets. Nevertheless, \cite{Mhaisen21} does not consider the energy consumption and latency of HFL. \cite{Luo20} proposes HFEL that solves both jointly solve the resource allocation and user assignment problems. \cite{Shengli22} proposes a joint bandwidth allocation and user assignment framework to improve the time delay and learning performance of training HFL. Our proposed method is compared with \cite{Luo20} and \cite{Shengli22} and exhibits the superior performance.

\section{System Model}\label{section3} 
A typical HFL framework includes $N$ users indexed by $\mathcal{N} = \left\{1,2,...,N\right\}$, $M$ edge servers indexed by $\mathcal{M} = \left\{1,2,...,M\right\}$, and a cloud server. Each mobile user $n$ has a local dataset $\mathcal{D}_n$ with $D_n$ data samples. The models trained on the local devices are called local models, while the models aggregated at the edge and cloud are called edge models and global models, respectively. 

\subsection{HFL training}
At the $i$-th global iteration, the training process of HFL contains three steps: local training, edge aggregation, and global aggregation.

1) \textbf{Local Training}: We define $\bm{w}_n^{i,k,l}$ as the model parameters of local model $n$. Local training adopts a gradient descent algorithm to update the local models as follows:
\begin{equation}
    \bm{w}_{n}^{i,k,l+1} = \bm{w}_{n}^{i,k,l}- \beta\nabla F_n(\bm{w}_{n}^{i,k,l}),
\label{localtrain}
\end{equation}
where $F_n(\bm{w}_n^{i,k,l})$ is the loss function of local model $n$ at the $k$-th edge iteration of the $i$-th global iteration, and $\beta$ is the learning rate. Local training ends when achieving the maximum number of local iterations $L$. 

2) \textbf{Edge Aggregation}: To carry out edge aggregation, the mobile users upload their model weights to the edge servers. We define $\Psi = \left[\mathcal{N}_{m}\big|m\in\mathcal{M} \right]$ as an user assignment pattern, where $\mathcal{N}_{m}\ (m = 1,...,M)$ is the set of mobile users assigned to edge server $m$. Each edge conducts aggregation by averaging the model weights:
\begin{equation}
    \bm{w}_m^{i,k+1} = \frac{\sum_{n \in \mathcal{N}_{m}}D_n\bm{w}_n^{i,k,L}}{D_{\mathcal{N}_{m}}},
\label{edgeaggr}
\end{equation}
where $D_{\mathcal{N}_{m}} = \sum_{n \in \mathcal{N}_{m}}D_n$. The edge will transmit the averaged model weights back to the mobile users. Local training and edge aggregation are repeated until the maximum edge iteration number $K$ is achieved. 

3) \textbf{Global Aggregation}: In this phase, the edge servers upload the weights to the cloud server. The cloud server receives the model weights from the edges and averages them:
\begin{equation}
    \bm{w}^{i+1} = \frac{\sum_{m=1}^MD_{\mathcal{N}_{m}}\bm{w}_m^{i,K}}{D},
\label{cloudaggr}
\end{equation}
where $D = \sum_{m=1}^MD_{\mathcal{N}_{m}}$. The cloud server broadcasts the global model to the scheduled users through the edge servers. Algorithm~\ref{HFL_training} explains the training process of HFL at the $i$-th global iteration. Note that one global iteration contains $K$ edge iterations (i.e., $L\times K$ local iterations). The convergence proof of the HFL framework is provided in our online file~\cite{proof} due to space limitation.

\begin{algorithm}[t]
\textsl{}\setstretch{1.05}
\caption{HFL training at the $i$-th global iteration.}  
\label{HFL_training}
\begin{algorithmic}[1] 
\Require{Global model $\bm{w}^i$, set of mobile users $\mathcal{N} = \left\{1,2,...,N\right\}$, set of edge servers $\mathcal{M} = \{1,...,M\}$, user assignment pattern $\Psi = \left[\mathcal{N}_{m}\big| m \in\mathcal{M}\right]$, maximum local iteration $L$, maximum edge iteration $K$}
\Ensure{Global model $\bm{w}^{i+1}$}
\State Initialize all the local models $\left\{\bm{w}_n^{i,0,0}\Big|n\in\mathcal{N}\right\}$ using $\bm{w}^{i}$
\For{$k = 1$ to $K$}
    \For{each edge server $m \in \mathcal{M}$ in parallel}
		\For{each mobile user $n \in \mathcal{N}_{m}$ in parallel}
			\State Perform local training based on~\eqref{localtrain} for $L$ iterations to obtain $\bm{w}_n^{i,k,L}$
			\State Send $\bm{w}_n^{i,k,L}$ to the edge server $m$
		\EndFor
		\State Performs edge aggregation based on~\eqref{edgeaggr} to obtain a new $\bm{w}_m^{i,k+1}$
	\EndFor
\EndFor
\State All the edge servers transmit the edge models to the cloud server
\State The cloud server performs global aggregation based on~\eqref{cloudaggr} to obtain $\bm{w}^{i+1}$;
\State \Return $\bm{w}^{i+1}$.
\end{algorithmic}
\end{algorithm}

\begin{figure*}[t]
  \centering
  \includegraphics[width=17cm]{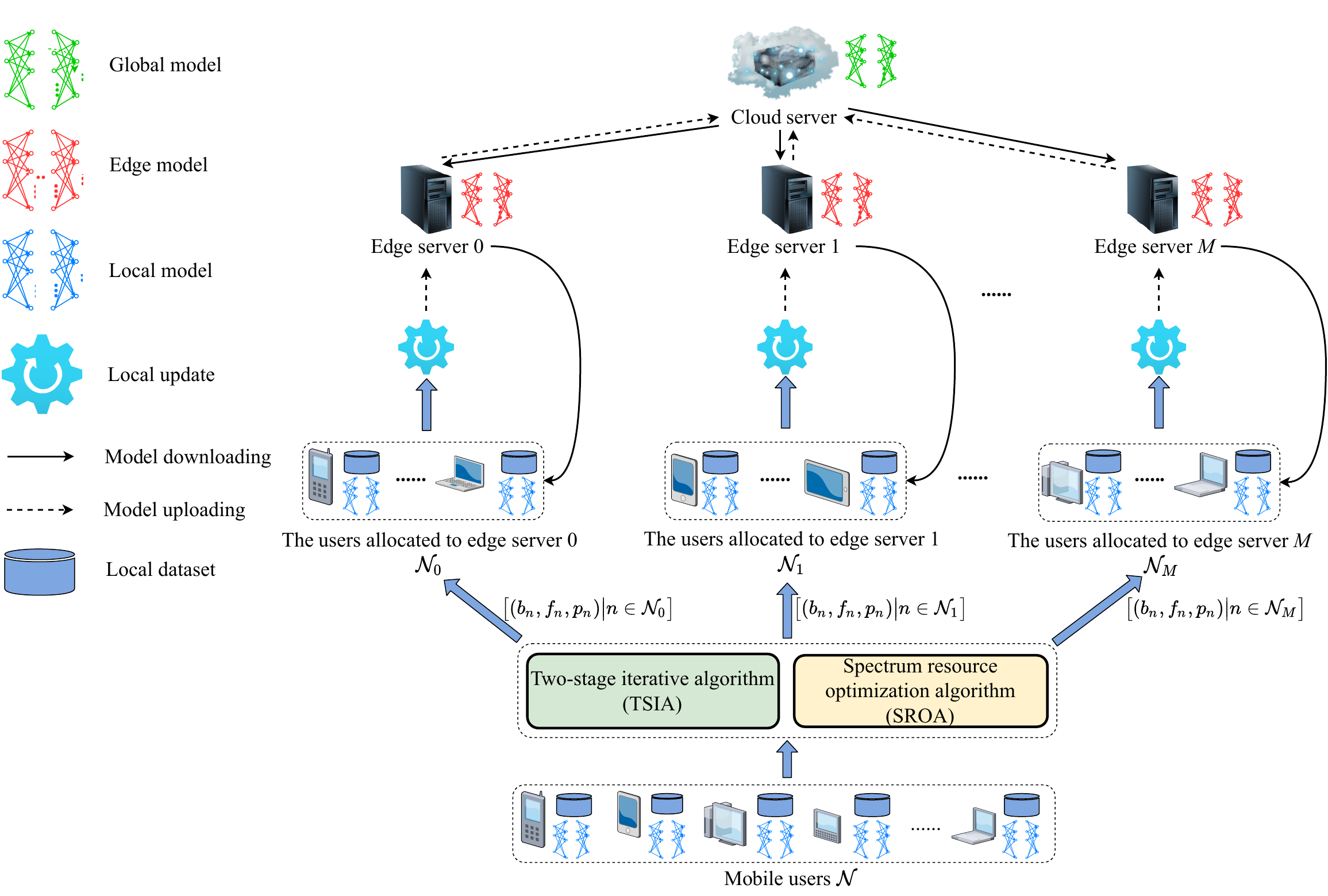}
\caption{Overview of the proposed HFL framework. Given a user assignment pattern $\Psi$, SROA obtains the optimal solution $\left\{(B_m, b_n, f_n, p_n) \big| n\in\mathcal{N}, m\in\mathcal{M} \right\}$ and calculates the objective value~\eqref{YY}. With SROA, TSIA searches for the optimal user assignment pattern which significantly reduces the objective value~\eqref{YY}. After finishing user assignment and resource allocation, HFL begins the training process according to Algorithm~\ref{HFL_training}.}
\label{FL_Framework}
\end{figure*}

\subsection{Energy consumption and time delay in HFL}
The energy consumption and time delay of HFL can be obtained by analyzing the computation and transmission processes.

We define $c_n$ as the number of CPU cycles for user $n$ to compute one data sample. Let $f_n$ be the CPU frequency of $n$. We assume that all data samples have the same size. At each edge iteration, the energy consumption $E_{n}^{\text{cmp}}$ and time delay $T_{n}^{\text{cmp}}$ for user $n$ to process $D_n$ data samples are:
\begin{equation}
    T_{n}^{\text{cmp}}= \frac{Lc_nD_n}{f_n},
\label{time_cmp}
\end{equation}
\begin{equation}
    E_{n}^{\text{cmp}} = \frac{\alpha}{2}L f_n^2 c_n D_n,
\label{energy_cmp}
\end{equation}
where $L$ is the maximum number of local iterations, and $\frac{\alpha}{2}$ denotes the effective capacitance coefficient of the user's computing chipset.

After $L$ local iterations, the mobile users will transmit the model weights to the corresponding edge server. In this work, we adopt a frequency-division multiple access (FDMA) protocol for data transmission. The achievable transmission rate of user $n$ is:
\begin{equation}
    r_n = b_n\text{log}_2(1+\frac{g_n^m p_n}{N_0b_n}),
\end{equation}
where $b_n$ is the bandwidth allocated to user $n$, $g_n^m$ is the channel gain between the mobile user $n$ and edge server $m$, $N_0$ is background noise, and $p_n$ is the transmit power. At each edge iteration, the energy consumption $E^{\text{com}}_{n}$ and time delay $T^{\text{com}}_{n}$ for user $n$ to transmit the model weights to the edge are:
\begin{equation}
    T^{\text{com}}_{n} = \frac{s}{r_n},
\end{equation}
\begin{equation}
    E^{\text{com}}_{n} = p_{n}T^{\text{com}}_{n},
\end{equation}
where $s$ denotes the size of the HFL model. Note that the time delay and energy consumption for mobile users to download the aggregated model weights are neglected compared with the upload mode. Based on the above discussion, the energy consumption $E_{m}$ and time delay $T_{m}$ of edge server $m$ after $K$ edge iterations are:
\begin{equation}
    T_{m} = K\max_{n \in \mathcal{N}_{m}}\left(T_{n}^{\text{cmp}} + T^{\text{com}}_{n} \right).
\end{equation}
\begin{equation}
    E_{m} = K\sum_{n \in \mathcal{N}_{m}}(E_{n}^{\text{cmp}} + E^{\text{com}}_{n}).
\label{Em}
\end{equation}

The aggregated model weights at each edge server are uploaded to the cloud server after $K$ edge iterations. The energy consumption $E^{\text{cloud}}_{m}$ and time delay $T^{\text{cloud}}_{m}$ for edge server $m$ to upload the model weights to the cloud server are derived as:
\begin{equation}
    T^{\text{cloud}}_{m} = \frac{s}{r_m},
\end{equation}
\begin{equation}
    E^{\text{cloud}}_{m} = p_mT^{\text{cloud}}_{m},
\end{equation}
where $s_m$ is the model size at edge server $m$, $r_m$ and $p_m$ are the transmission rate and the average transmit power of edge server $m$ for model uploading, respectively. 

The aggregation time at the edge and cloud servers is neglected as the time is much smaller than the time delay on the mobile users. As HFL cannot conduct global aggregation until the cloud server receives the model weights from all the edge servers, the time delay of HFL is determined by the slowest edge. Therefore, we have the following equations:
\begin{equation}
    T = \max_{m \in \mathcal{M}}\left( T^{\text{cloud}}_{m}+T_{m} \right), T_\text{sum} = I\cdot T,
\end{equation}
\begin{equation}
    E = \sum_{m \in \mathcal{M}}(E^{\text{cloud}}_{m} + E_{m}), E_\text{sum} = I\cdot E,
\label{E_total}
\end{equation}
where $I$ is the number of global iterations, $T$ and $E$ are the time delay and energy consumption of HFL at each global iteration, respectively. $T_\text{sum}$ and $E_\text{sum}$ are the time delay and energy consumption of training the entire HFL algorithm, respectively.\vspace{-10pt} 

\subsection{Problem formulation}
In this paper, the optimization problem is formulated to balance the tradeoff between the energy cost and the time delay:
\begin{align}
&\underset{\emph{\textbf{B}}, \emph{\textbf{b}}, \emph{\textbf{f}}, \emph{\textbf{p}}, \Psi}{\text{minimize}} \quad R = E_{\text{sum}} + \lambda T_{\text{sum}} \label{YY} \\
&\text{s.t.} \quad \sum_{n \in \mathcal{N}_{m}} b_n \leq B_m, \forall m \in \mathcal{M},  \tag{\ref{YY}a} \label{YYa} \\
&\quad \quad \sum_{m \in \mathcal{M}} B_m \leq B, \tag{\ref{YY}b} \label{YYb} \\
&\quad \quad 0 \leq f_n \leq f^{\text{max}}_n, \quad \forall n \in \mathcal{N}, \tag{\ref{YY}c} \label{YYc} \\
&\quad \quad 0 \leq p_n \leq p^{\text{max}}_n, \quad \forall n \in \mathcal{N}, \tag{\ref{YY}d} \label{YYd} \\
&\quad \quad \mathcal{N}_{\mu} \cap \mathcal{N}_{\nu} = \varnothing, \forall \mu, \nu \in \mathcal{M} \text{ and } \mu \neq \nu, \tag{\ref{YY}e} \label{YYe} \\
&\quad \quad \mathcal{N} = \bigcup_{m \in \mathcal{M}} \mathcal{N}_{m}, \tag{\ref{YY}f} \label{YYf}
\end{align}
where $\emph{\textbf{B}} = \left[B_m \big| m\in\mathcal{M} \right]$, $\bm{b} = \left[b_n \big| n\in\mathcal{N} \right]$, $\bm{f} = \left[f_n \big| n\in\mathcal{N} \right]$, $\bm{p} = \left[p_n \big| n\in\mathcal{N} \right]$, $\Psi = \left[\mathcal{N}_m \big| m\in\mathcal{M} \right]$. $B$ represents the total bandwidth resources for all the edge servers, and $B_m$ represents the bandwidth of edge server $m$. $f^{\text{max}}_n$ and $p^{\text{max}}_n$ represent the maximum CPU frequency and maximum average transmit power of user $n$, respectively. $\lambda$ is an importance weighting indicator. \eqref{YYa} and \eqref{YYb} denote the bandwidth constraints of the mobile users and edge servers, respectively. \eqref{YYc} and \eqref{YYd} denote the constraints for the CPU frequency and the transmission power, respectively. \eqref{YYe} ensures each mobile user only communicates with one edge server. \eqref{YYf} indicates that all the mobile users participate in the training process.

Problem~\eqref{YY} merges combinatorial optimization (i.e., the optimization of $\Psi$) with multi-variable optimization  (i.e., the optimization of $\emph{\textbf{B}}, \emph{\textbf{b}}, \emph{\textbf{f}}$, and $\emph{\textbf{p}}$). Therefore, it is intractable to directly obtain the global optimum of problem~\eqref{YY}. In this paper, we decompose problem~\eqref{YY} into a resource allocation problem and a user assignment problem and propose a novel HFL framework as shown in Fig.~\ref{FL_Framework} to attain the locally optimal solution to problem~\eqref{YY}. Specifically, given a known user assignment pattern $\Psi$, $\left\{(B_m, b_n, f_n, p_n) \big| n\in\mathcal{N}, m\in\mathcal{M} \right\}$ as well as the objective value $R$ can be obtained by a spectrum resource optimization algorithm (SROA). With SROA, a two-stage iterative algorithm (TSIA) is designed to find the optimal user assignment pattern $\Psi^\ast$. After performing SROA and TSIA, the proposed framework starts to train the HFL model based on Algorithm~\ref{FL_Framework}.

\section{Spectrum Resource Optimization Algorithm}\label{resource_allo}
We introduce the spectrum resource optimization algorithm (SROA) method before the user assignment method as the proposed user assignment method in this paper relies on SROA. Given $\Psi$, SROA aims to minimize the weighted sum of time delay and energy consumption for training the entire HFL algorithm. Based on \eqref{time_cmp}-\eqref{YY}, the optimization problem is formulated as follows:
\begin{align}
&\underset{\emph{\textbf{B}}, \emph{\textbf{b}}, \emph{\textbf{f}},\emph{\textbf{p}}}{\text{min}}\  E_{\text{sum}} + \lambda T_{\text{sum}} \notag \\
&=  I \Big( \sum_{m\in\mathcal{M}}  \Big(\sum_{n\in\mathcal{N}_m}\Big(  \frac{K p_n s}{b_n\log_2(1+\frac{g_n^m p_n}{N_0 b_n})}+\frac{\alpha}{2}f_n^2 K L c_n D_n \Big)+ \notag \\
&  E^{\text{cloud}}_{m}\Big)+ \lambda \max_{m\in\mathcal{M}}\Big(\max_{n\in\mathcal{N}_m}\Big( \frac{K s}{b_n\log_2(1+\frac{g_n^m p_n}{N_0 b_n})} + \frac{L K c_n D_n}{f_n }\Big) \notag \\
&  + T^{\text{cloud}}_{m} \Big) \Big)\label{XX}\\
&\text{s.t.} \quad \sum_{n \in \mathcal{N}_{m}} b_n \leq B_m, \forall m \in \mathcal{M},  \tag{\ref{XX}a} \label{XXa} \\
&\quad \ \; \, \sum_{m \in \mathcal{M}} B_m \leq B, \tag{\ref{XX}{b}} \label{XXb}\\
&\quad\ \;\; \,\,  0 \leq f_n \leq f^{\text{max}}_n,\ \forall n \in \mathcal{N}, \tag{\ref{XX}{c}} \label{XXc}\\
&\quad\ \;\; \,\,  0 \leq p_n \leq p^{\text{max}}_n,\ \forall n \in \mathcal{N}. \tag{\ref{XX}{d}} \label{XXd}
\end{align}

Obtaining the globally optimal solution for problem~\eqref{XX} poses significant challenges as multiple variables are required to be optimized in problem~\eqref{XX}. Besides, the objective function~\eqref{XX} involves a double summation (i.e., $\sum\limits_{m\in\mathcal{M}}\sum\limits_{n\in \mathcal{N}_{m}}$) and a double maximization (i.e., $\max\limits_{m\in\mathcal{M}}\max\limits_{n\in\mathcal{N}_m}$), thus further complicating problem~\eqref{XX}. To simplify the optimization problem, we introduce four variables $\psi_{n,m}$, $h_n$, $\delta_n$, and $t$ to rewrite problem~\eqref{XX} as follows:
\begin{align}
&\underset{\emph{\textbf{b}}, \emph{\textbf{f}},\emph{\textbf{p}}}{\text{min}}\ \sum_{n\in\mathcal{N}}\left( \frac{I K p_n s}{b_n\log_2(1+\frac{h_n p_n}{N_0 b_n})}  +\frac{\alpha}{2}f_n^2 I K L c_n D_n \right) +  \lambda t \label{ZZ} \\
& \text{s.t.} \quad \sum_{n \in \mathcal{N}} b_n \leq B,  \tag{\ref{ZZ}a} \label{ZZa} \\
&\quad\ \;\; \,\,  0 \leq f_n \leq f^{\text{max}}_n,\ \forall n \in \mathcal{N}, \tag{\ref{ZZ}{b}} \label{ZZb}\\
& \quad\ \;\; \,\,  0 \leq p_n \leq p^{\text{max}}_n,\ \forall n \in \mathcal{N}, \tag{\ref{ZZ}{c}} \label{ZZc} \\
& \quad\ \;\; \,\,   \frac{I K s}{b_n\log_2(1+\frac{h_n p_n}{N_0 b_n})} + \frac{I L K c_n D_n}{f_n} + \delta_n \geq t,\ \forall n \in \mathcal{N}, \tag{\ref{ZZ}{d}} \label{ZZd}
\end{align}
where $\psi_{n,m} = 1$ if mobile user $n$ is assigned to edge server $m$, and $\psi_{n,m} = 0$, otherwise, $h_n = \sum\limits_{m\in\mathcal{M}}\psi_{n,m}\ g^m_n$, $\delta_n =I \sum\limits_{m\in\mathcal{M}}\psi_{n,m}\ T^{\text{cloud}}_{m}$, and $t = \max\limits_{n\in\mathcal{N}}\Big( \frac{I K s}{b_n\log_2(1+\frac{h_n p_n}{N_0 b_n})} + \frac{I L K c_n D_n}{f_n} + \delta_n \Big)$. Besides, $\sum\limits_{m\in\mathcal{M}}\sum\limits_{n\in \mathcal{N}_{m}} E^{\text{cloud}}_{m}$ is a constant and can be omitted from the objective function. With $h_n$ and $\delta_n$, the $\sum\limits_{m\in\mathcal{M}}\sum\limits_{n\in \mathcal{N}_{m}}(\cdot)$ and $\max\limits_{m\in\mathcal{M}}\max\limits_{n\in\mathcal{N}_m}(\cdot)$ terms in problem~\eqref{XX} can be modified as $\sum\limits_{n\in\mathcal{N}}(\cdot)$ and $\max\limits_{n\in\mathcal{N}}(\cdot)$ in problem~\eqref{ZZ}, respectively. In addition, the bandwidth constraints~\eqref{XXa} and~\eqref{XXb} are merged into a single constraint~\eqref{ZZa}. Consequently, the resource allocation problem between the mobile users and multiple edge servers is transformed into a resource allocation problem between the mobile users and a single server with a total bandwidth resource of $B$, which facilitates the optimization procedure. Note that after solving problem~\eqref{ZZ}, $B_m$ can be obtained by $\sum\limits_{n\in\mathcal{N}_m} b_n$.

To solve problem~\eqref{ZZ}, SROA sets $t$ to a fixed value, and the optimal solutions of $\emph{\textbf{b}}, \emph{\textbf{f}},\text{ and }\emph{\textbf{p}}$ under the fixed $t$ are obtained. Next, SROA adopts a binary search algorithm to find the optimal $t^\ast$ for the problem~\eqref{ZZ}. SROA iteratively performs the above two steps and finally converges to a locally optimal solution $(\emph{\textbf{b}}^\ast, \emph{\textbf{f}}^\ast, \emph{\textbf{p}}^\ast)$.

%     T^{\text{cloud}}_{m} = \frac{s_m}{r_m},
% \end{equation}
% \begin{equation}
%     E^{\text{cloud}}_{m} = p_mT^{\text{cloud}}_{m},

\subsection{Optimal $\textbf{p}, \textbf{b},\text{and }\textbf{f}$ with a fixed $t$}
For simplifying the presentation, we first provide some notations in this work. 
\begin{align}
    A_n &= \frac{\alpha}{2}IKLc_nD_n,J_n = IKLc_nD_n,H_n = IKs,\label{W}\\
    U_n &= IKsp_n,Y_n = \frac{IKs}{b_n},X = \sum_{n\in\mathcal{N}}\frac{\alpha}{2}f_n^2 I K L c_n D_n,\label{WW} \\
    G_n &= \frac{p_n h_n}{N_0},Z_n = \frac{h_n}{N_0b_n},F_n = \frac{I L K c_n D_n}{f_n} + \delta_n. \label{WWWW}
\end{align}

\begin{algorithm}[t]
\textsl{}\setstretch{1.05}
\caption{Obtain optimal ($\emph{\textbf{b}},\emph{\textbf{f}}$) with fixed ($\emph{\textbf{p}}, t$)}  
\label{bf_opti}
\begin{algorithmic}[1] 
\Require{$B$, $\emph{\textbf{p}}$, $t$, $f^{\text{max}}_n$, $\varepsilon_0$, $b_\text{max}$}
\Ensure{Optimal CPU frequency $\emph{\textbf{f}}$ and bandwidth $\emph{\textbf{b}}$}
\State $f_n^{\text{low}} = \max(0, \frac{J_n}{t-\delta_n - \ln 2\cdot H_n/G_n})$, $f_n^{\text{up}} = f^{\text{max}}_n$, $\forall n \in \mathcal{N}$
\While {$\max\limits_{n\in\mathcal{N}}\left(\frac{f_n^\text{up} -f_n^\text{low} }{f_n^\text{up}}\right) > \varepsilon_0$}
    \State $f_n = \frac{f_n^{\text{low}}+f_n^{\text{up}}}{2}$, $\forall n \in \mathcal{N}$
    \State Apply a bisection method with an upper bound $b_\text{max}$ to calculate $\big[b_n\big|n \in$ $\mathcal{N}\big]$ using \eqref{Xa}
    \State $b_\text{sum} = \sum\limits_{n \in \mathcal{N}}b_n$
    \If{$b_\text{sum} < B$}
        \State $f_n^{\text{up}} = f_n$, $\forall n \in \mathcal{N}$
    \ElsIf{$b_\text{sum} > B$}
        \State $f_n^{\text{low}} = f_n$, $\forall n \in \mathcal{N}$
    \Else
        \State \textbf{break}
    \EndIf
\EndWhile
\State \Return $\emph{\textbf{b}} = [b_n| n\in \mathcal{N}]$, $\emph{\textbf{f}} = [f_n| n\in \mathcal{N}]$
\end{algorithmic}
\end{algorithm}

Given a fixed $t$, the optimization goal of problem~\eqref{ZZ} becomes minimizing the energy consumption. In this case, we take two steps to obtain the optimal $\emph{\textbf{p}}, \emph{\textbf{b}},\text{and }\emph{\textbf{f}}$. Firstly, we design a binary search algorithm to obtain the optimal $\emph{\textbf{b}},\text{and }\emph{\textbf{f}}$ with fixed $(t, \emph{\textbf{p}})$. Secondly, we propose another binary search algorithm to find the optimal $\emph{\textbf{p}}$.

For the first step, the optimization problem~\eqref{ZZ} is reformulated as follows:
\begin{align}
&\underset{\emph{\textbf{b}}, \emph{\textbf{f}}}{\text{min}}  \sum_{n \in \mathcal{N} }\Big(A_nf_n^2+\frac{U_n}{b_n\log_2(1+\frac{G_n}{b_n})}\Big)\label{X}\\
&\textrm{s.t.} \quad \frac{H_n}{b_n\log_2(1+\frac{G_n}{b_n})}+\frac{J_n}{f_n} + \delta_n \leq  t, \forall n \in \mathcal{N}  \tag{\ref{X}{a}} \label{Xa}\\
&\quad\ \; \, \sum_{n \in \mathcal{N}}b_n \leq B, \tag{\ref{X}{b}} \label{Xb}\\
&\quad\ \;\; \,\,  0 \leq f_n \leq f^{\text{max}}_n, \forall n \in \mathcal{N}.\tag{\ref{X}{c}} \label{Xc}
\end{align}
On the one hand, given fixed ($\emph{\textbf{p}}, t$), the energy consumption in~\eqref{X} becomes smaller with the increasing of $b_n$ and the decreasing of $f_n$. On the other hand, it can be seen from \eqref{Xa} that reducing $f_n$ increases $b_n$. Therefore, the key idea for solving~\eqref{X} is to minimize $f_n$ while satisfying the constraint~\eqref{Xb}. To this end, we provide Lemma~\ref{b_proof} to facilitate the optimization. With Lemma~\ref{b_proof}, we propose a binary search method to optimize $(\bm{b}, \bm{f})$ as shown in Algorithm~\ref{bf_opti}. In Algorithm~\ref{bf_opti}, $\varepsilon_0$ is the tolerance of the algorithm. The upper bound of $f_n$ is initialized as $f_n^\text{max}$, while the lower bound is derived based on Lemma~\ref{b_proof}. In each iteration, $f_n$ is updated as the average of its upper and lower bounds. In Line 4, the bisection method is employed using \eqref{Xa} to calculate $b_n$ with the known $f_n$ since $b_n\log_2(1+\frac{G_n}{b_n})$ is a monotonically increasing function with respect to $b_n$. If the sum of $b_n$ (i.e., $b_\text{sum}$) is lower than $B$, it indicates that HFL can achieve the time delay $t$ while preserving bandwidth resources of $(B - b_\text{sum})$. In such cases, $f_n$ can be further reduced to minimize energy consumption. Conversely, if $b_\text{sum}$ exceeds $B$, $f_n$ is increased. Algorithm~\ref{bf_opti} reaches convergence when the maximum value of $\frac{f_n^\text{up} -f_n^\text{low} }{f_n^\text{up}}$ becomes smaller than the tolerance $\varepsilon_0$.

\begin{lemma}
$b_n\log_2(1+\frac{G_n}{b_n})$ is a monotonically increasing function w.r.t. $b_n$ with an upper bound $\frac{G_n}{\ln2}$. The lower bound of $f_n$ is $\max(0, \frac{J_n}{t-\delta_n - \ln 2\cdot H_n/G_n})$.
\label{b_proof}
\end{lemma}
\begin{proof}
See Appendix A.
\end{proof}

\begin{algorithm}[t]
\textsl{}\setstretch{1.05}
\caption{Obtain optimal $\emph{\textbf{p}}$ with fixed $t$}  
\label{p_opti}
\begin{algorithmic}[1]
\Require{$B$, $t$, $p^{\text{max}}_n$}
\Ensure{Optimal transmit power $\emph{\textbf{p}}$}
\State Initialize $p^{\text{up}}_n= p^{\text{max}}_n$, $p^{\text{low}}_n = \max(0, \zeta(2^{\frac{\gamma}{\eta}}-1))$, where $\gamma = \frac{IKs}{b_\text{max}}$, $\eta = t - \delta_n - \frac{IKLc_nD_n}{f_n^\text{max}}$, and $\zeta = \frac{N_0 b_\text{max}}{h_n}$, $\forall n \in \mathcal{N}$
\While {$\max\limits_{n\in\mathcal{N}}\left( \frac{p^{\text{up}}_n-p^{\text{low}}_n}{p^{\text{up}}_n} \right) > \varepsilon_1$}
    \State $p_n = \frac{p^{\text{up}}_n+p^{\text{low}}_n}{2}$, $\forall n \in \mathcal{N}$
    \State Obtain $(\emph{\textbf{b}},\emph{\textbf{f}})$ based on Algorithm~\ref{bf_opti}
    \State $b_\text{sum} = \sum\limits_{n \in \mathcal{N}}b_n$
    \If{$b_\text{sum} < B$}
        \State $p^{\text{up}}_n = p_n$
    \ElsIf{$b_\text{sum} > B$}
        \State $p^{\text{low}}_n = p_n$
    \Else
        \State \textbf{break}
    \EndIf  
\EndWhile
\State  \Return $\emph{\textbf{p}} = [p_n| n\in \mathcal{N}]$
\end{algorithmic}
\end{algorithm}
%b和f也是固定的
For the second step, we can investigate how to optimize $\emph{\textbf{p}}$ given the fixed $(\bm{b}, \bm{f}, t)$. The optimization problem~\eqref{ZZ} is rewritten as follows:
\begin{align}
&\underset{\emph{\textbf{p}}}{\text{min}}  \sum_{n \in \mathcal{N} }\Big(\frac{Y_np_n}{\log_2(1+Z_np_n)}\Big) + X \label{Y}\\
&\textrm{s.t.} \quad \frac{Y_n}{\log_2(1+Z_np_n)}+F_n \leq  t, \forall n \in \mathcal{N}  \tag{\ref{Y}{a}} \label{Ya}\\
&\quad\ \;\; \,\,  0 \leq p_n \leq p^{\text{max}}_n, \forall n \in \mathcal{N}.\tag{\ref{Y}{b}} \label{Yb}
\end{align}
There is a positive correlation between $\emph{\textbf{p}}$ and objective function~\eqref{Y}. Besides, reducing $\emph{\textbf{p}}$ increases $\emph{\textbf{f}}$ and $\emph{\textbf{b}}$ given the fixed $t$ according to \eqref{ZZd}. Therefore, Algorithm~\ref{p_opti} is proposed as a binary search method to find the optimal $\emph{\textbf{p}}$. Firstly, Algorithm~\ref{p_opti} initializes the upper and lower bounds of $p_n$ according to $p_n^\text{max}$ and Lemma~\ref{p_proof}, respectively. Next, $\emph{\textbf{p}}$ is set to be the average of its upper and lower bounds. Then, $(\emph{\textbf{b}},\emph{\textbf{f}})$ is obtained using Algorithm~\ref{bf_opti}. Finally, $p^{\text{up}}_n$ will be reduced if $b_\text{sum}$ is smaller than $B$, otherwise $p^{\text{low}}_n$ is increased. The iteration stops when the tolerance $\varepsilon_1$ is met.

\begin{lemma}
The lower bound of $p_n$ is $\max(0, \zeta(2^{\frac{\gamma}{\eta}}-1))$, where $\gamma = \frac{IKs}{b_\text{max}}$, $\eta = t - \delta_n - \frac{IKLc_nD_n}{f_n^\text{max}}$, and $\zeta = \frac{N_0 b_\text{max}}{h_n}$.
\label{p_proof}
\end{lemma}
\begin{proof}
See Appendix B.
\end{proof}

\begin{algorithm}[t]
\textsl{}\setstretch{1.05}
\caption{Optimal solutions of the problem~\eqref{XX}}  
\label{T_opti}
\begin{algorithmic}[1] 
\Require{$\lambda$, $t^{\text{up}}$, $t^{\text{low}}$}
\Ensure{Optimal spectrum resource allocation}
\State $R^\ast = +\infty $
\While {$\frac{t^{\text{up}}-t^{\text{low}}}{t^{\text{up}}} > \varepsilon_2$}
    \State $t = \frac{t^{\text{low}}+t^{\text{up}}}{2}$
    \State Obtain ($\emph{\textbf{b}}$, $\emph{\textbf{f}}$, $\emph{\textbf{p}}$) using Algorithms~\ref{bf_opti} and~\ref{p_opti} with $t$
    \If{$b_\text{sum} > B$}
        \State $t^{\text{low}} = t$
        \State $\textbf{Continue}$
    \EndIf
    \State Calculate $E_\text{sum}$ based on~\eqref{E_total};
    \State $R = E_\text{sum} + \lambda t$
    \If{$R > R^\ast$ }
        \State $t^{\text{low}} = t$
    \Else
        \State $t^{\text{up}} = t$
        \State $R^\ast = R$
    \EndIf  
\EndWhile 
\State $t^\ast = t$;
\State Obtain $(R^\ast,\emph{\textbf{b}}^\ast, \emph{\textbf{f}}^\ast, \emph{\textbf{p}}^\ast)$ using Algorithms~\ref{bf_opti} and~\ref{p_opti} with $t^\ast$
\State $B^\ast_m = \sum\limits_{n\in\mathcal{N}_m} b_n$
\State \Return $(R^\ast, \emph{\textbf{B}}^\ast, \emph{\textbf{b}}^\ast, \emph{\textbf{f}}^\ast, \emph{\textbf{p}}^\ast)$.
\end{algorithmic}
\end{algorithm}

\subsection{Binary search algorithm for obtaining $t^\ast$}
A binary search algorithm is employed to obtain the optimal $t^\ast$ as shown in Algorithm~\ref{T_opti}. To begin with, the upper and lower bounds of $t$ are set as a large enough number and a small enough number, respectively. The optimal objective value $R^\ast$ is initialized as $+\infty$. At each iteration, $t$ is set as the average of its upper and lower bounds. Given $t$, $(\emph{\textbf{b}}, \emph{\textbf{f}}, \emph{\textbf{p}})$ is derived via Algorithms~\ref{bf_opti} and Algorithm~\ref{p_opti}, and the objective value~\eqref{XX} is obtained with the derived $(\emph{\textbf{b}}, \emph{\textbf{f}}, \emph{\textbf{p}})$. $t$ will be reduced if a smaller $R$ is obtained, and vice versa. Note that the sum of $b_n$ (i.e. $b_\text{sum}$) can be larger than $B$ if $t$ becomes too small. Therefore, at Lines 5$-$8 of Algorithm~\ref{T_opti}, we calculate $b_\text{sum}$ after carrying out Algorithm~\ref{bf_opti} and~\ref{p_opti}, and $t$ will be increased if $b_\text{sum} > B$. The solutions of~\eqref{XX} (i.e., $\emph{\textbf{B}}^\ast, \emph{\textbf{b}}^\ast, \emph{\textbf{f}}^\ast, \emph{\textbf{p}}^\ast$) can be obtained after $t^\ast$ is found.

\subsection{Complexity analysis}
To begin with, the complexity of Algorithm~\ref{bf_opti} is $\mathcal{O}\left( N^2 \log_2{\left(\frac{1}{\varepsilon_0}\right)}\log_2{\left(\frac{1}{\varepsilon_3}\right)} \right)$, where $\varepsilon_3$ is the tolerance for calculating $b_n$ at Line 5 via the bisection method. Algorithm~\ref{p_opti} involves complexity $\mathcal{O}\Big( N^3 \log_2{\left(\frac{1}{\varepsilon_0}\right)}\log_2{\left(\frac{1}{\varepsilon_1}\right)}\log_2{\left(\frac{1}{\varepsilon_3}\right)} \Big)$. As a result, the complexity of Algorithm~\ref{T_opti} is $\mathcal{O}\Big( N^3 \log_2{\left(\frac{1}{\varepsilon_0}\right)}\log_2{\left(\frac{1}{\varepsilon_1}\right)}\log_2{\left(\frac{1}{\varepsilon_2}\right)}\log_2{\left(\frac{1}{\varepsilon_3}\right)} \Big)$. Algorithms~\ref{bf_opti}$-$\ref{T_opti} are all established by binary search algorithms. Therefore, Algorithm~\ref{T_opti} will finally converge to a stable point and derive a locally optimal solution. According to the experimental results in Section~\ref{srm}, the derived solution of Algorithm~\ref{T_opti} outperforms existing baselines.

\section{User Assignment Optimization via a Two-stage Iterative Algorithm}\label{edge_allo}

Below we first present the preliminaries of the two-stage iterative algorithm (TSIA) in Section~\ref{Criterion}, and then explain TSIA details in Section~\ref{overview-TSIA}. 

\subsection{Preliminaries of TSIA}\label{Criterion}
Before elaborating TSIA, we introduce several important definitions. First of all, we define $R_m$ as the weighted sum of the energy consumption and time delay of edge server $m$:\vspace{3pt}
\begin{equation}
    R_m = (E^{\text{cloud}}_{m} + E_{m}) + \lambda \left( T^{\text{cloud}}_{m}+T_{m} \right). \vspace{3pt}
\end{equation}
Then, we define \emph{\textbf{costly server}}, \emph{\textbf{economic server}}, \emph{\textbf{costly user}}, and \emph{\textbf{economic user}}. According to Definition~\ref{serverDefinition}, the costly and economic servers are the edge servers with the highest and the lowest $R_m$, respectively. According to Definition~\ref{B-TSIA}, for edge server $m \in \mathcal{M}$, the user that utilizes the most bandwidth resources of the edge server is regarded as the costly user, while the user allocated the least bandwidth is the economic user. Generally, each edge server has only one costly user and one economic user.

\begin{definition}
In TSIA, the costly edge $m^+$ and economic edge $m^-$ are defined as
\begin{align}
m^+ = \underset{m \in \mathcal{M}}{\emph{argmax}}(R_m) \\
m^- = \underset{m \in \mathcal{M}}{\emph{argmin}}(R_m)
\end{align}
\label{serverDefinition}
\end{definition}
\begin{definition}
In TSIA, the costly user $n_m^+$ and economic user $n_m^-$ of edge server $m \in \mathcal{M}$ are defined as
\begin{align}
n_m^+ = \underset{n \in \mathcal{N}_{m}}{\emph{argmax}}(b_n) \\
n_m^- = \underset{n \in \mathcal{N}_{m}}{\emph{argmin}}(b_n)
\end{align}
\label{B-TSIA}
\end{definition}

\begin{algorithm}[!t]
\textsl{}\setstretch{1.05}
\caption{Two-stage iterative algorithm (TSIA)}
\label{firter}
\begin{algorithmic}[1] 
\Require{$\mathcal{N}$, $\mathcal{M}$}
\Ensure{Optimal user assignment pattern $\Psi^\ast$}
\State $stage = 1$, $q = 1$
\While{$stage\leq 2$}
    \If{$stage = 1$}\Comment{The first stage}
        \State Initialize $\Psi_q$ based on geographical distance
        \State $\Psi^\ast = \Psi_q$
        \State Obtain $\left[(R, R_m, B_m, b_n,f_n, p_n) \big| n\in\mathcal{N}, m\in\mathcal{M} \right]$ via Algorithm~\ref{T_opti}
        \State $R^\ast = R$
    \Else\Comment{The second stage}
        \State $\Psi_q = \Psi^\ast$
    \EndIf
    \Repeat
        \State $m^+ = \underset{m \in \mathcal{M}}{\mathrm{argmax}}(R_{m}),\ m^- = \underset{m \in \mathcal{M}}{\mathrm{argmin}}(R_{m})$
        \If{$stage = 1$}
            \State Assign the costly user $n^+_{m^+}$ from $\mathcal{N}_{m^+}$ $\mathcal{N}_{m^-}$
        \Else
            \State Assign the economic user $n^-_{m^+}$ from $\mathcal{N}_{m^+}$ to $\mathcal{N}_{m^-}$
        \EndIf
        \State Obtain $\left[(R, R_m, B_m, b_n,f_n, p_n) \big| n\in\mathcal{N}, m\in\mathcal{M} \right]$ via Algorithm~\ref{T_opti}
        \If{$R < R^\ast$}
            \State $\Psi^\ast = \Psi_q$
            \State $R^\ast = R$
        \EndIf
        \State $q = q+1$
    \Until{the algorithm achieves convergence}
    \State $stage = stage + 1$
\EndWhile
\State \Return $\Psi^\ast$, $\left[(B_m, b_n,f_n, p_n) \big| n\in\mathcal{N}, m\in\mathcal{M} \right]$
\end{algorithmic}
\end{algorithm}

\subsection{Details of TSIA}\label{overview-TSIA}
The core idea of TSIA is to bridge the performance gap between the costly server and the economic server by transferring the user from the costly server to the economic server. The costly user will be transferred at the first stage, while the second stage aims to transfer the economic user. The mechanism of TSIA is straightforward. When the costly user $n_m^+$ is transferred from edge server $m$ to other servers, edge server $m$ regains large amounts of bandwidth. When SROA is conducted before and after such a transfer, the results of SROA vary noticeably. In contrast, the economic user $n_m^-$ does not consume many bandwidth resources of edge server $m$. The transfer of $n_m^-$ can be viewed as a fine-tuning process as it has a minor impact on the total system cost. TSIA results in a balanced user assignment pattern that significantly reduces the objective value~\eqref{YY}.

% Specifically, a user will be assigned to the edge server that is closest to itself. This assignment method can result in a highly imbalanced assignment pattern (some edge servers are assigned with few users while some other are burdened with many mobile users), hence leading to straggler servers and deteriorating the overall performance of HFL. From the perspective of initialization, however, the geographical distance-based assigning method enables TSIA to converge to the stable point that significantly reduces the objective value~\eqref{YY}.

Algorithm~\ref{firter} describes the implementation details of TSIA. We define $\Psi_q$ ($q\in \mathbb{Z}^+$) as the user assignment pattern at the $q$-th assigning step. In the first stage of TSIA (i.e., $stage = 1$), the user assignment pattern $\Psi_q$ and $\Psi^\ast$ are initialized based on geography distance (Line 5). Then, SROA is carried out given the initial user assignment pattern to initialize the optimal objective value $R^\ast$ (Lines 6-7). Next, TSIA iteratively identify the costly server $m^+$, economic server $m^-$, and the costly user of the costly server (i.e., $n^+_{m^+}$) based on Definition~\ref{serverDefinition} and~\ref{B-TSIA}; the costly user $n^+_{m^+}$ is transmitted to edge server $m^-$  (Lines 12-14). The second stage (i.e., $stage = 2$) shares many same steps with the first stage. The differences are that the second stage initializes the user assignment pattern using the pattern derived in the first stage (Line 9), and the economic user $n^-_{m^+}$ is transferred from the costly server $m^+$ to the economic server $m^-$ (Line 16).

\begin{remark}
TSIA is a deterministic policy. Specifically, given $\Psi_q$, the result of Algorithm~\ref{T_opti}, the determination of $\left\{m^+, m^-, n^+_m, n^-_m\right\}$, and the action taken by TSIA (Line 14 or Line 16 in Algorithm~\ref{firter}) are all deterministic. Therefore, with a fixed $\Psi_q$, TSIA always derives the same $\Psi_{q+1}$.
\label{deterministic}
\end{remark}

We propose Remark~\ref{deterministic} to facilitate the convergence justification of TSIA. In addition to Remark~\ref{deterministic}, the number of combinations of user assignment patterns is $M^N$ given the numbers of mobile users $N$ and edge servers $M$. As a result, if TSIA continues assigning the mobile users after $M^N$ steps, there must exist one unique assignment pattern $\Psi_q (q\in \mathbb{Z}^+$ and $q\in[1,M^N])$ that satisfies $\Psi_{M^N+1} = \Psi_q$, which indicates the assigning pattern returns to the $q$-th assigning step. Theoretically, the complexity of TSIA is $\mathcal{O}(M^N)$. Empirically, however, TSIA requires much fewer assigning iterations to reach convergence according to Fig.~\ref{change}. For example, given $N=50$ and $M=5$, the assigning iterations for HFL to converge roughly ranges from 20 to 50, which is notably lower than $5^{50}$.

% \begin{remark}
% Given $N$ mobile users and $M$ edge servers, there are totally $M^N$ types of assignment patterns $\Psi$. 
% \label{totalpattern}
% \end{remark}

% \begin{theorem}
% In the worst case, TSIA requires $2M^N$ assigning steps to reach convergence (i.e., $q \leq 2M^N$).
% \label{convergence_prof}
% \end{theorem}

% \begin{proof}
% See Appendix C.
% \end{proof}

\section{Performance Evaluation}\label{section6}
\subsection{Experimental setup}\label{setup}

We consider $N = 50$ mobile users and $M = 5$ edge users randomly distributed in a square of side 500 m, and the centre of the square is the cloud server. The total bandwidth resources of the edge servers are randomly distributed in $[10,1000]\ \text{KHz}$. The path loss model is $128.1 + 37.6\log_{10}d$(km), and the standard deviation of shadow fading is $8$dB. The power spectrum density of the additive Gaussian noise is {$N_0 = -174$ dBm/MHz.} 

We set an equal maximum CPU frequency $f^\text{max}_1 = ... = f^\text{max}_N = 5\ \text{GHz}$ and an equal maximum transmit power $p^\text{max}_1 = ... = p^\text{max}_N = 23\ \text{dBm}$. $c_n\ (n \in \mathcal{N})$ is randomly distributed in $[1,10] \times 10^4$ cycles/sample. The effective capacitance coefficient $\alpha = 2\times10^{-28}$. The maximum local iteration $L$ and maximum edge iteration $K$ are both 5 in this paper. The dataset for training HFL and the network structure of the HFL model are provided as follows.

\textbf{FashionMNIST}: FashionMNIST~\cite{xiao2017} is a dataset that contains Zalando's article images. It consists of 60000 training samples and 10000 testing samples. Each data sample is a $28 \times 28$ grey-scale image. There are a total of 10 classes in FashionMNIST. The HFL model used for training FashionMNIST consists of two $5\times5$ convolution layers, and the output dimension of the first convolution layer and the second convolution layer are 10 and 12, respectively. The convolution layers are followed by $2\times2$ max pooling. The output of the last pooling layer is flattened and fed into a linear layer. The size of the local dataset  $D_n$ is randomly distributed in $[1000,1400]$. The size of the model weights $s = 446\ \text{KB}$. 

\textbf{CIFAR-10}: CIFAR-10~\cite{Krizhevsky09} is a ten-class dataset that contains 60000 $32 \times 32$ color images. Each class has 5000 training images and 1000 testing images. The HFL model has two $5\times5$ convolution layers and two linear layers. The output channels of the two convolution layers are 10 and 20, respectively. Each convolution layer is followed by $2\times2$ max pooling. The output of the last pooling layer is connected with two linear layers. The size of the local dataset  $D_n$ is randomly distributed in $[800,1200]$. The size of the model weights $s = 523\ \text{KB}$. 

\textbf{ImageNette}: ImageNette~\cite{Howard20} is formulated by selecting ten easily classified classes from ImageNet~\cite{Olga15}. It includes 9469 training samples and 3925 testing samples. We resize the images to $32 \times 32$. The HFL model contains two $5\times5$ convolution layers and two linear layers. The output channels of the two convolution layers are 15 and 28, respectively. The convolution layers are followed by $2\times2$ max pooling. The output dimensions of the linear layers are 300 and 10, respectively. The size of the local dataset $D_n$ is randomly distributed in $[150,220]$. The size of the model weights $s = 881\ \text{KB}$.

Among these datasets, the images in FashionMNIST are the simplest and easiest to be identified, while ImageNette has an incremental level of difficulty over FashionMNIST and CIFAR-10.

\subsection{Results of resource optimization methods}\label{srm}
We compare the proposed SROA with several state-of-art baselines. Although SROA is primarily designed for FDMA systems, its applicability can be readily extended to alternative communication schemes such as Orthogonal Frequency Division Multiple Access (OFDMA). Therefore, we assess the performance of SROA in both FDMA and OFDMA scenarios. For the baselines in the FDMA scheme, JDSRA~\cite{Shi21} is used for comparison. For the baselines in the OFDMA scheme, ERA~\cite{Zeng20} and the resource allocation method in JUARA~\cite{Shengli22} are used as the baseline methods. In addition, FEDL~\cite{Dinh20} and the resource allocation method in HFEL~\cite{Luo20} can be applied for both FDMA scheme and OFDMA scheme. The experiments are conducted under $D_n \in [150,220]$, $s = 881\ \text{KB}$, $I = 80$, and the mobile users are simply assigned to the edge server closest to their location.

\begin{figure}[!t]
\centering
\subfigure[FDMA scheme]{
\begin{minipage}[t]{1\linewidth}
\centering
\includegraphics[width=6cm]{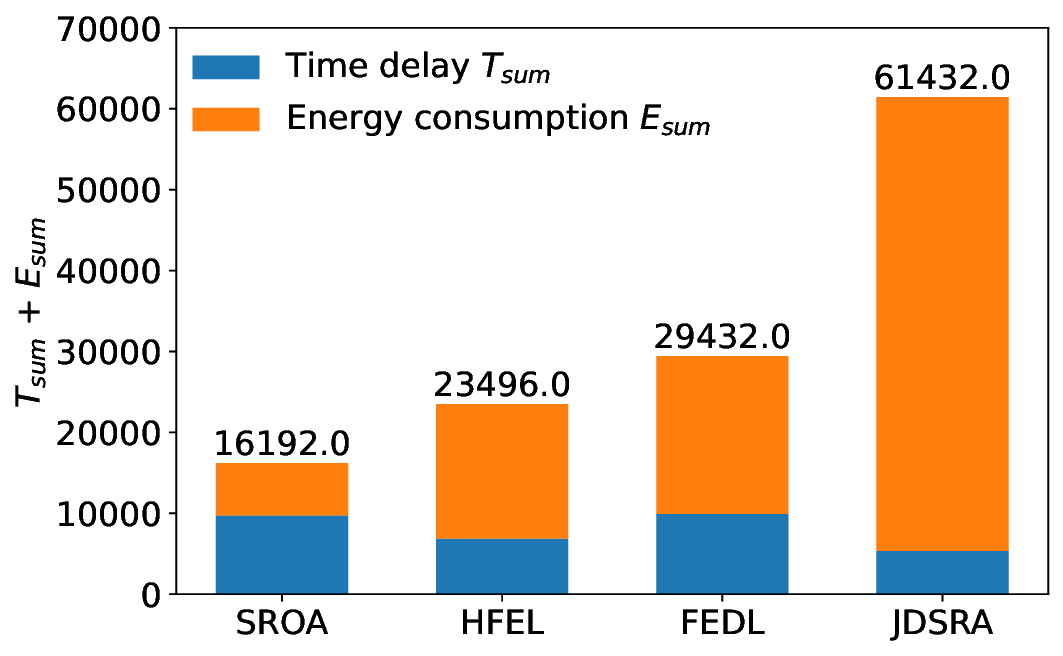}
%\caption{fig1}
\label{FDMA_result}
\end{minipage}%
}%

\subfigure[OFDMA scheme]{
\begin{minipage}[t]{1\linewidth}
\centering
\includegraphics[width=6cm]{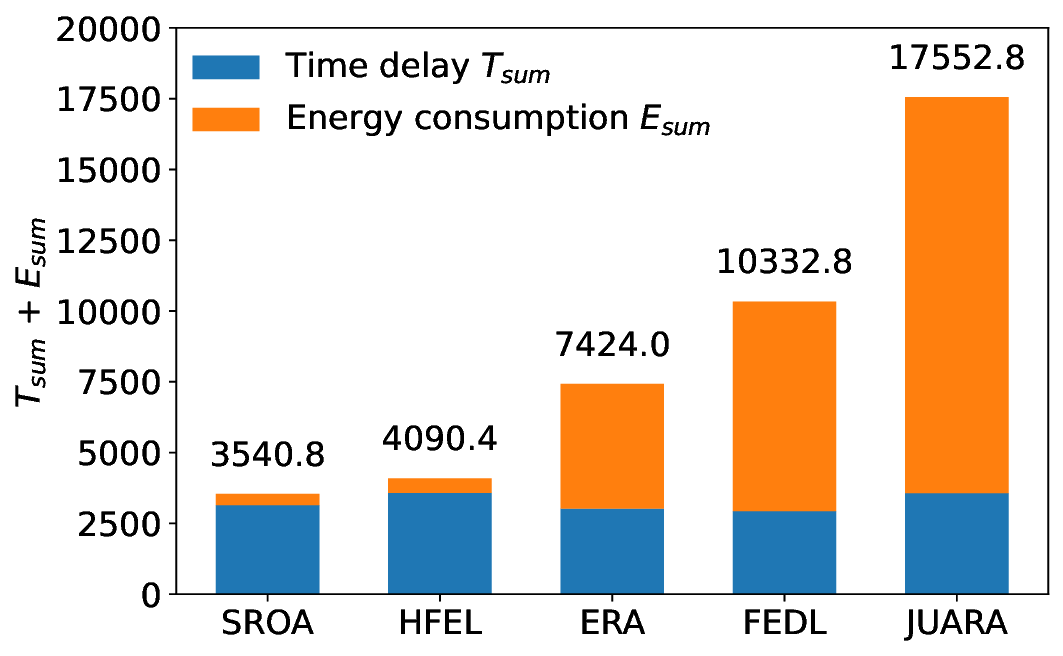}
\label{OFDMA_result}
%\caption{fig2}
\end{minipage}%
}%
\caption{Objective value~\eqref{YY} using different resource allocation methods. For HFEL, FEDL, and the proposed SROA, the importance weight $\lambda = 1$.}
\label{FDMA_OFDMA}
\end{figure}

\begin{figure}[!t]
\centering
\subfigure[FDMA scheme]{
\begin{minipage}[t]{1\linewidth}
\centering
\includegraphics[width=6cm]{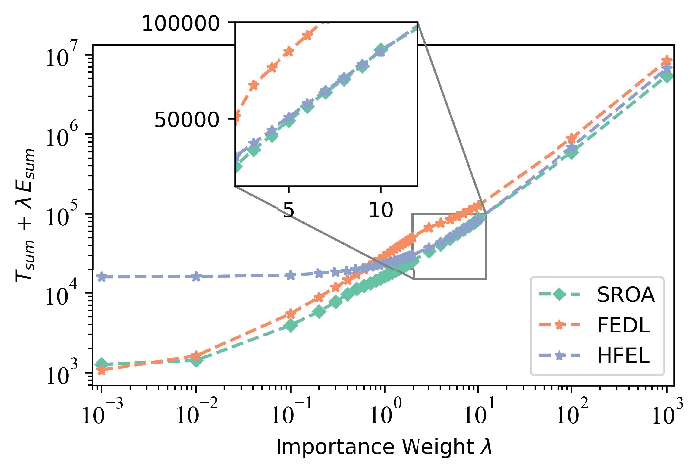}
%\caption{fig1}
\label{obj_FDMA}
\end{minipage}%
}%

\subfigure[OFDMA scheme]{
\begin{minipage}[t]{1\linewidth}
\centering
\includegraphics[width=6cm]{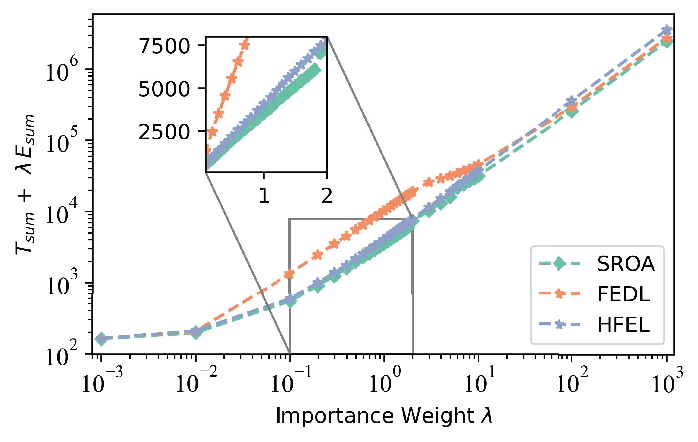}
\label{obj_OFDMA}
%\caption{fig2}
\end{minipage}%
}%
\caption{Objective value~\eqref{YY} using SROA, HFEL, and FEDL with different $\lambda$.}
\label{obj_FDMA_OFDMA}
\end{figure}

Figure~\ref{FDMA_OFDMA} illustrates the objective value~\eqref{YY} obtained through various resource allocation methods. In our analysis, we assign equal importance weights ($\lambda=1$) to our proposed method (SROA), FEDL, and HFEL, indicating an equal emphasis on time delay and energy cost. From the results depicted in Figure~\ref{FDMA_OFDMA}, it is evident that SROA achieves the lowest objective value in both the FDMA and OFDMA schemes. This outcome can be attributed to the comprehensive nature of SROA, which considers multiple factors such as time delay and energy consumption, and optimizes the allocation of bandwidth, CPU frequency, and average transmit power in a joint manner. In contrast, existing baselines either optimize fewer parameters or considers only one perspective (either energy consumption or time delay). Figure~\ref{obj_FDMA_OFDMA} displays the objective value~\eqref{YY} obtained through different resource allocation methods for various values of $\lambda$. The range of $\lambda$ spans from $10^{-3}$ to $10^3$, and a logarithmic scale is employed on the y-axis to enhance observation. It can be noted from Figure~\ref{obj_FDMA_OFDMA} that SROA consistently achieves the lowest objective values across different $\lambda$ values, with the exception of the FDMA scheme when $\lambda = 10$. Additionally, by examining the enlarged subfigures, it can be observed that the interval between adjacent tick marks is 50000 in Fig.\ref{obj_FDMA} and 2500 in Fig.\ref{obj_OFDMA}, respectively. Although the curves of SROA and HFEL appear visually close to each other, the objective value of SROA is significantly lower than that of HFEL. This observation further demonstrates the superiority of SROA over HFEL and FEDL in terms of achieving lower objective values.

\begin{figure}[t]
  \centering
  \includegraphics[width=7cm]{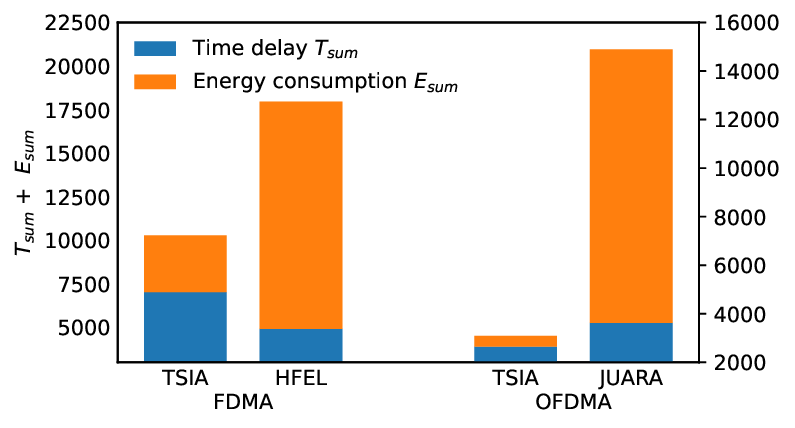}
\caption{Objective value~\eqref{YY} using different user assignment methods. For the proposed TSIA and HFEL, the importance weight $\lambda = 1$.}
\label{FDMA_OFDMA_assign}
\end{figure}

\begin{figure*}[t]
\centering
\begin{minipage}[t]{0.5\linewidth}
\centering
\includegraphics[width=6cm]{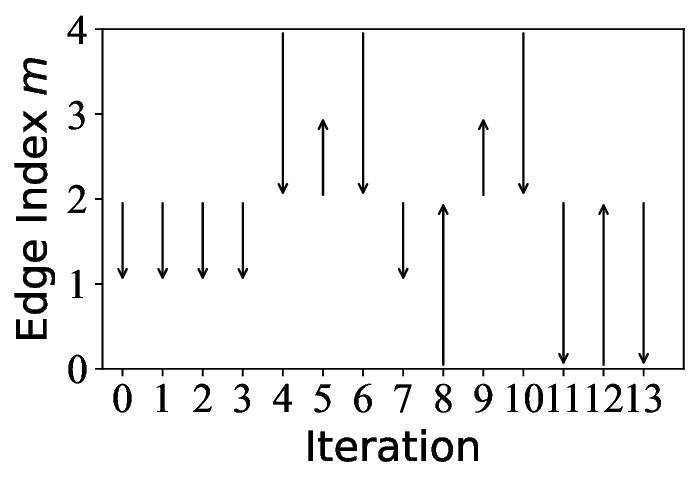}
\end{minipage}%
\begin{minipage}[t]{0.5\linewidth}
\centering
\includegraphics[width=6cm]{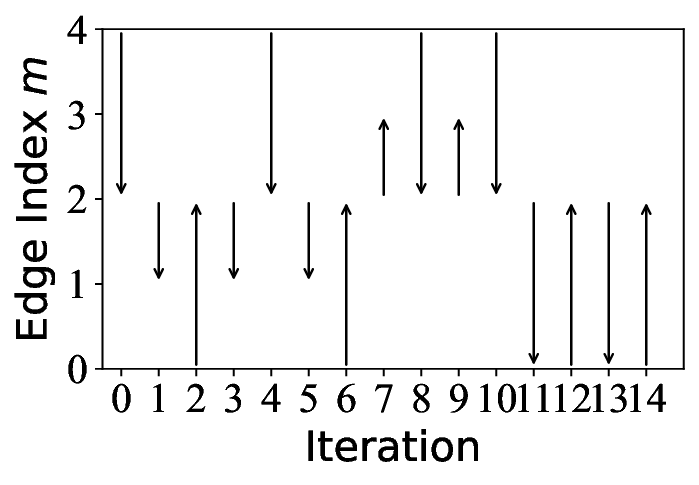}
\end{minipage}%

\begin{minipage}[t]{0.5\linewidth}
\centering
\includegraphics[width=6cm]{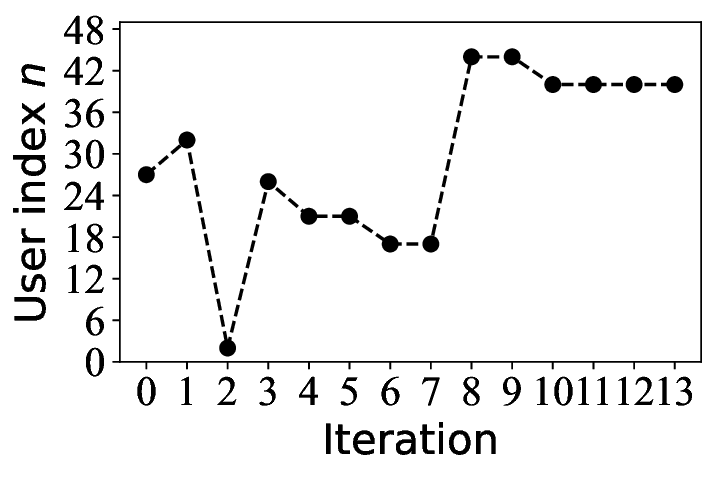}
\label{b_1}
\caption*{(a) TSIA (first stage)}
\end{minipage}%
\begin{minipage}[t]{0.5\linewidth}
\centering
\includegraphics[width=6cm]{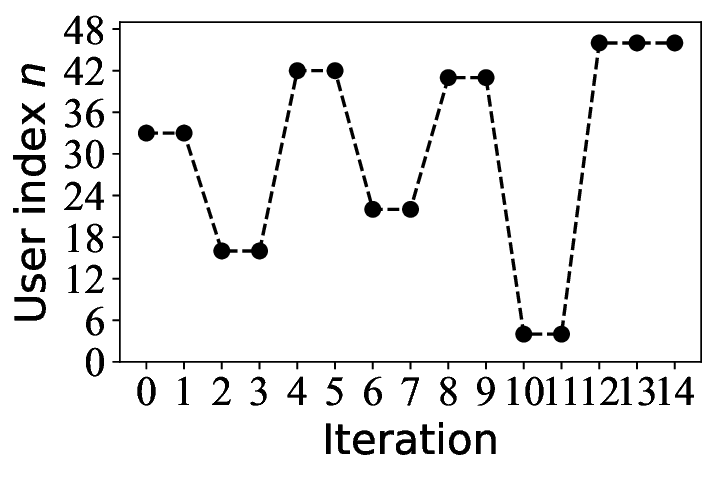}
\label{b_2}
\caption*{(b) TSIA (second stage)}
\end{minipage}\vspace{-0pt}
\caption{User assignment process of TSIA under $\lambda = 1$. The above two figures indicate the index of the edge server that transfers and receives the mobile user. The two figures below show the index of the transferred user. Take (a) as an example. At the 0-th iteration, mobile user 27 is transferred from edge server 2 to edge server 1. After the 12-th iteration, mobile user 40 is repeatedly transferred between edge server 0 and edge server 2. In this case, the first stage of TSIA uses 13 iterations to reach convergence.\vspace{-0pt}}
\label{TSIA-iter}
\end{figure*}

\begin{figure*}[!t]
\centering
\subfigure[Different number of mobile users $N$]{
\begin{minipage}[t]{0.5\linewidth}
\centering
\includegraphics[width=6cm]{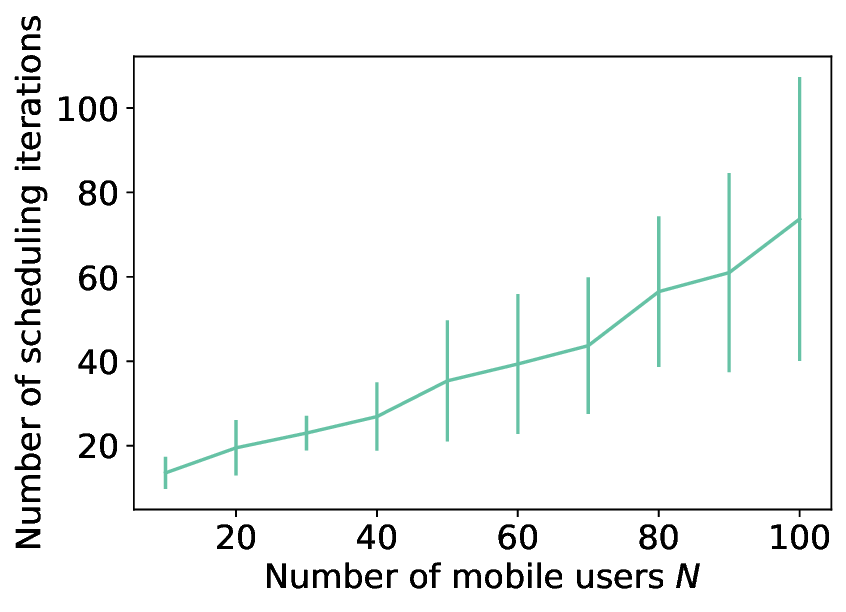}
%\caption{fig1}
\label{change_N}
\end{minipage}%
}%
\subfigure[Different number of edge servers $M$]{
\begin{minipage}[t]{0.5\linewidth}
\centering
\includegraphics[width=6cm]{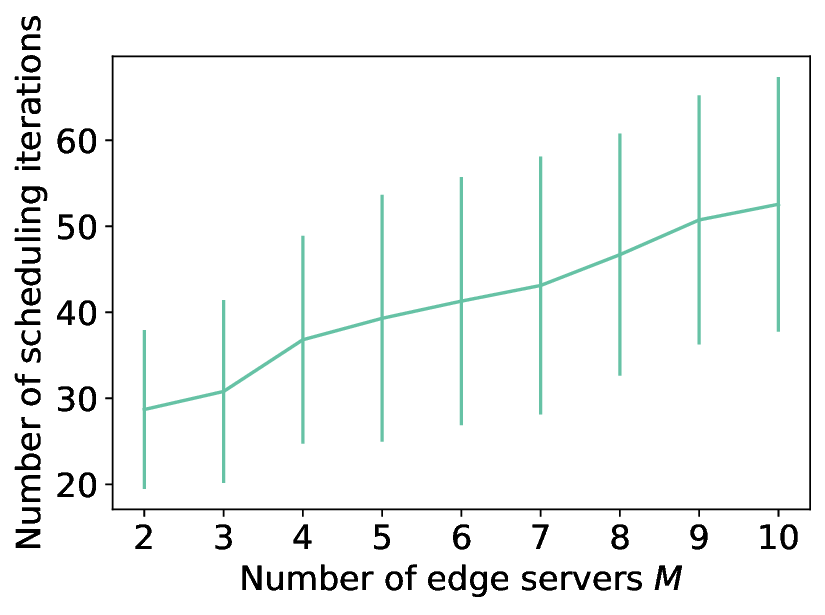}
\label{change_M}
%\caption{fig2}
\end{minipage}%
}%
\caption{Number of assigning iterations for TSIA to converge. The curve denotes the average result over five experiments, and the error bar represents the standard deviation.\vspace{-0cm}}
\label{change}
\end{figure*}

\subsection{Results of user assignment methods}\label{edge_result}
In this section, we evaluate the performance of the proposed user assignment algorithms. The experimental parameters are set as follows: $D_n \in [150,220]$, $s = 881\ \text{KB}$, $\lambda = 1$, and $I = 80$. We compare the performance of TSIA with two baseline methods, namely HFEL~\cite{Luo20} and JUARA~\cite{Shengli22}. HFEL adopts a random generation approach for the initial user assignment pattern and utilizes device transferring adjustment and device exchanging adjustment techniques. As HFEL requires a long execution time to reach convergence, we set the numbers of device transferring adjustment and device exchanging adjustment iterations to 100 and 300, respectively. JUARA addresses the user assignment problem by leveraging Lagrangian relaxation and derives the user assignment pattern using Karush-Kuhn-Tucker (KKT) conditions, given a large enough time delay. Next, JUARA iteratively reduces the time delay value by a fixed step until it becomes smaller than the lower bound of the time delay. In our experiments, the total number of assigning iterations for JUARA is set to 100, with the step length adjusted accordingly. For both TSIA and HFEL, we assign an importance weight of $\lambda = 1$.

\begin{figure*}[t]
\centering
\subfigure[FashionMNIST]{
\begin{minipage}[t]{0.33\linewidth}
\centering
\includegraphics[width=6cm]{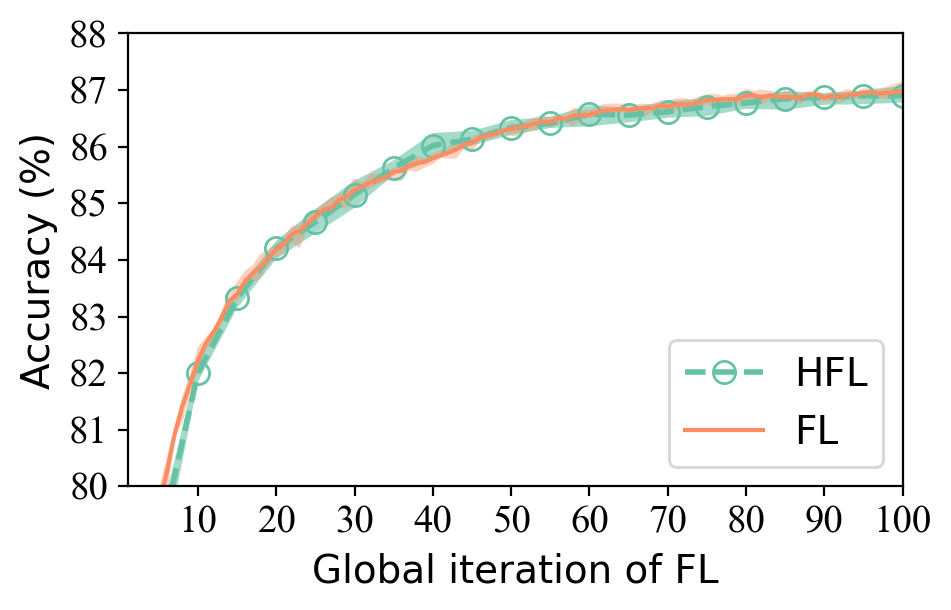}
%\caption{fig1}
\end{minipage}%
}%
\subfigure[CIFAR-10]{
\begin{minipage}[t]{0.33\linewidth}
\centering
\includegraphics[width=5.8cm]{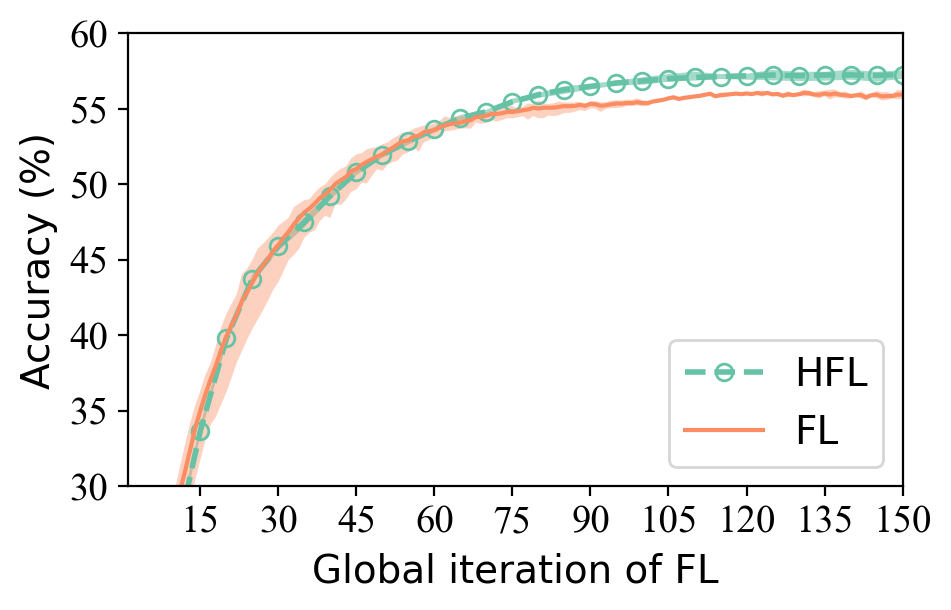}
%\caption{fig2}
\end{minipage}%
}%
\subfigure[ImageNette]{
\begin{minipage}[t]{0.33\linewidth}
\centering
\includegraphics[width=5.8cm]{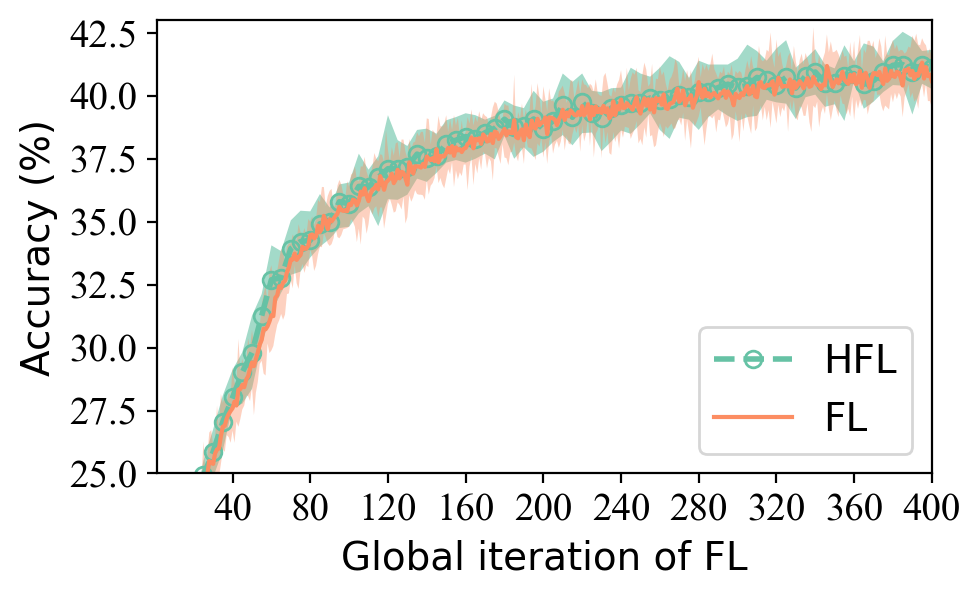}
%\caption{fig2}
\end{minipage}%
}%
\caption{Testing accuracy of HFL and FL on multiple datasets. Note that the global iteration of FL contains $L_\text{FL} = 5$ local iterations, while HFL requires $L\times K = 25$ local iterations for each global iteration. Therefore, one global iteration in HFL corresponds to five global iterations in FL. Moreover, the testing accuracy can only be obtained during global aggregation as the testing set is deployed at the cloud server. To show this phenomenon, the accuracy of HFL is represented by the discrete points.\vspace{-0cm}}
\label{datasets}
\end{figure*}

\begin{figure*}[h!]
\centering
\subfigure[FashionMNIST (FDMA)]{
\begin{minipage}[h]{0.33\linewidth}
\centering
\includegraphics[width=5.5cm]{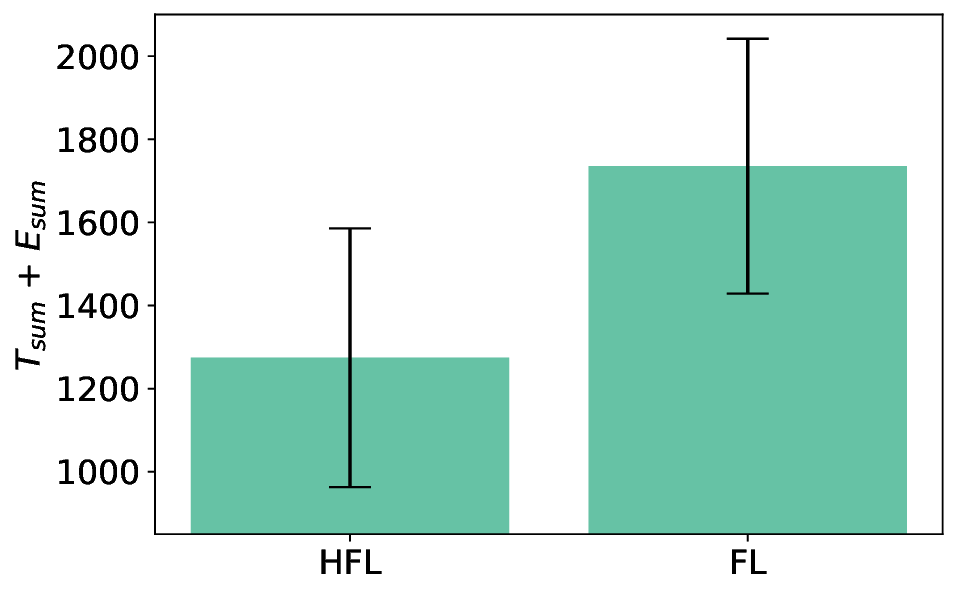}
%\caption{fig1}
\end{minipage}%
}%
\subfigure[CIFAR-10 (FDMA)]{
\begin{minipage}[h]{0.33\linewidth}
\centering
\includegraphics[width=5.5cm]{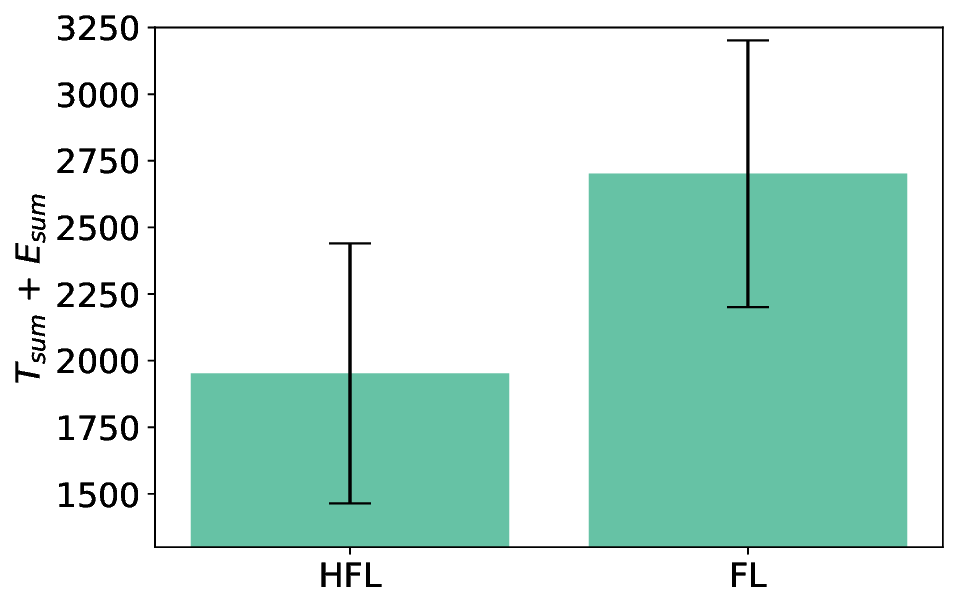}
%\caption{fig2}
\end{minipage}%
}%
\subfigure[ImageNette (FDMA)]{
\begin{minipage}[h]{0.33\linewidth}
\centering
\includegraphics[width=5.5cm]{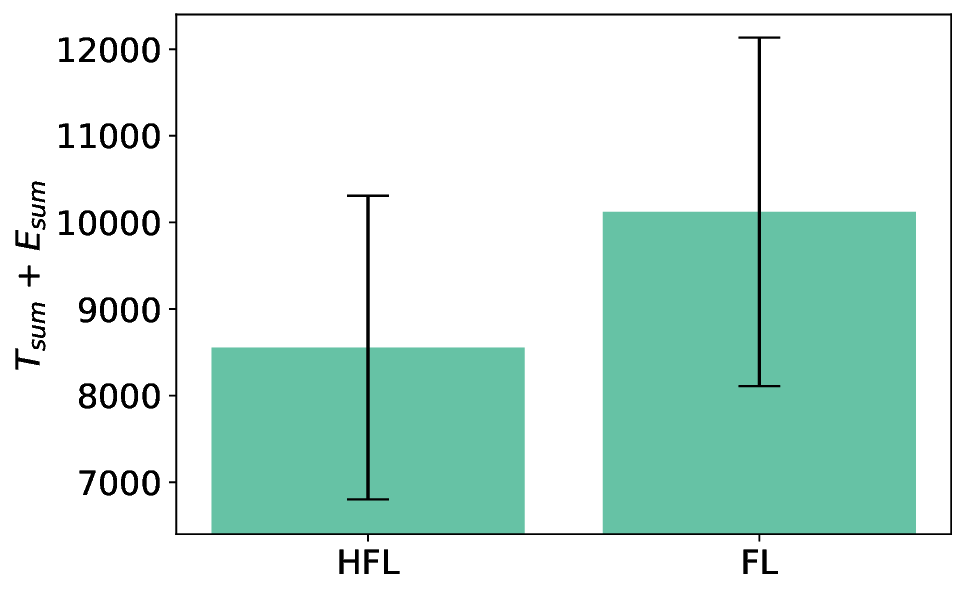}
%\caption{fig2}
\end{minipage}%
}%

\subfigure[FashionMNIST (OFDMA)]{
\begin{minipage}[h]{0.33\linewidth}
\centering
\includegraphics[width=5.5cm]{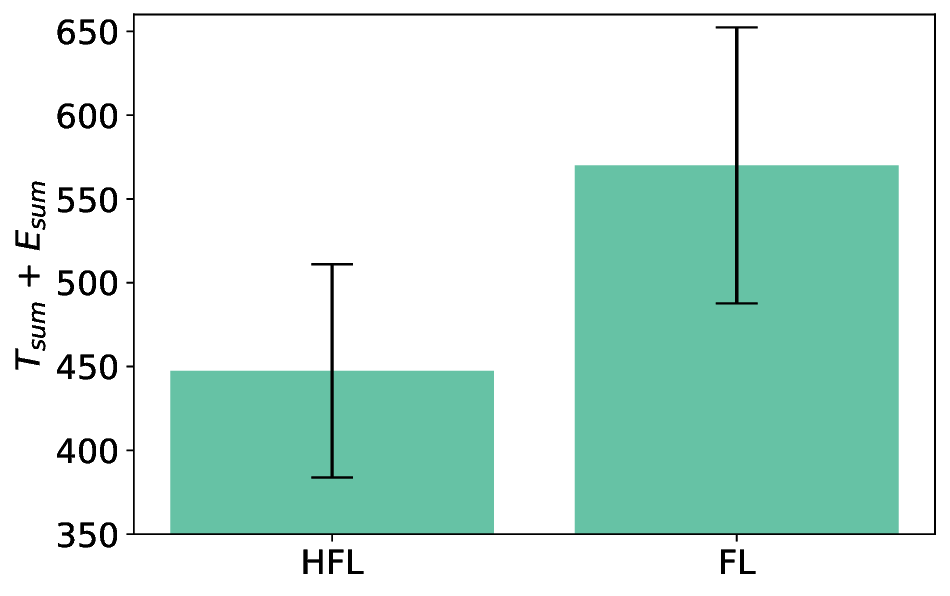}
%\caption{fig1}
\end{minipage}%
}%
\subfigure[CIFAR-10 (OFDMA)]{
\begin{minipage}[h]{0.33\linewidth}
\centering
\includegraphics[width=5.5cm]{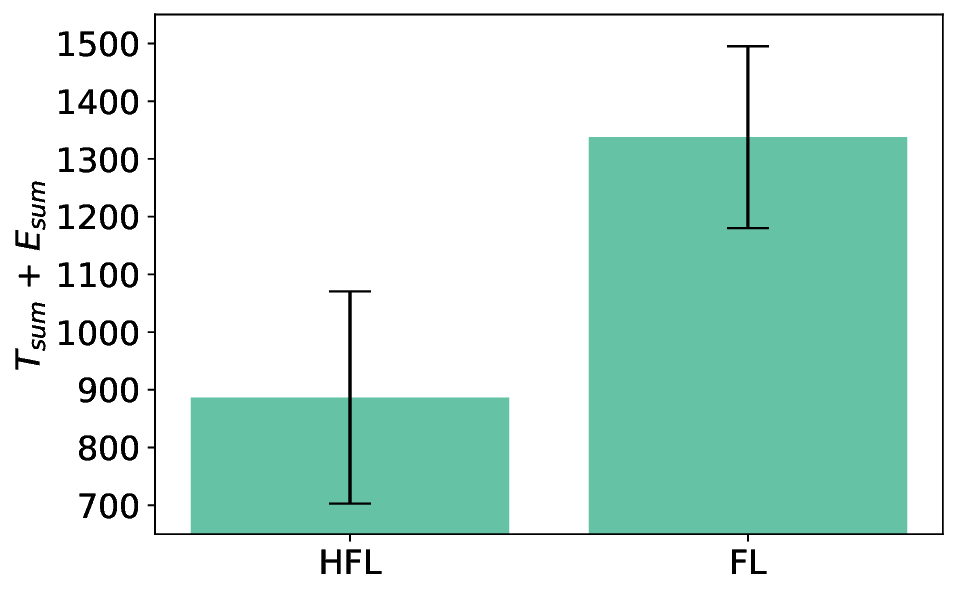}
%\caption{fig2}
\end{minipage}%
}%
\subfigure[ImageNette (OFDMA)]{
\begin{minipage}[h]{0.33\linewidth}
\centering
\includegraphics[width=5.5cm]{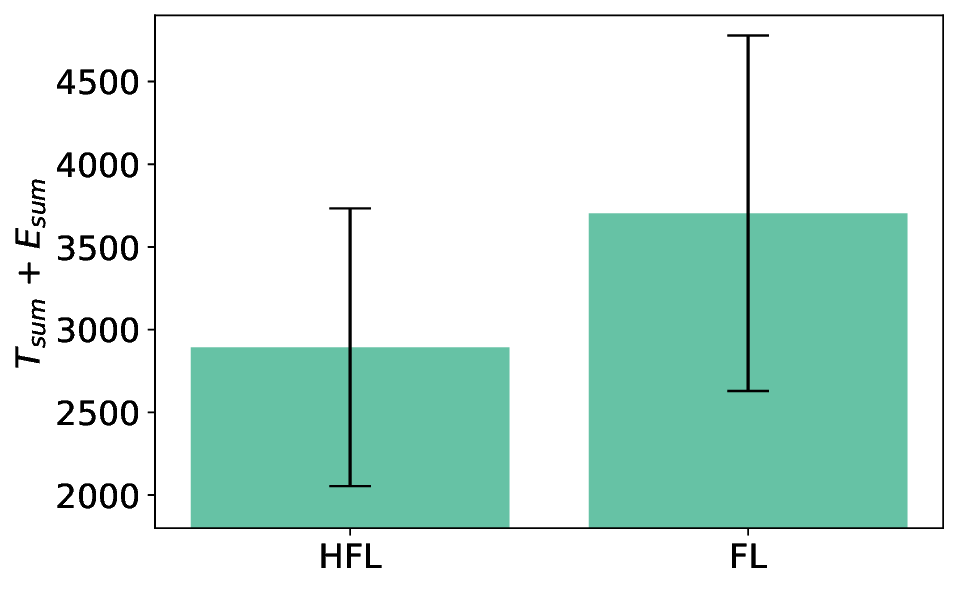}
%\caption{fig2}
\end{minipage}%
}%
\caption{Objective value~\eqref{YY} of HFL and FL on multiple datasets.}
\label{compare_FL}
\end{figure*}

Fig.~\ref{FDMA_OFDMA_assign} provides the objective values~\eqref{YY} under different user assignment methods. For the FDMA scheme, the time delay $T_\text{sum}$ is higher in TSIA than in HFEL. However, TSIA effectively reduces the energy consumption $E_\text{sum}$ of the HFL system, leading to a lower overall objective value. This improvement is attributed to TSIA's ability to transfer users from costly servers to economic servers, thereby achieving a balanced workload distribution. In contrast, HFEL relies on random adjustments, making it challenging to achieve an optimal assignment pattern, particularly in large search spaces. For the OFDMA scheme, TSIA outperforms JUARA in terms of both time delay and energy consumption. JUARA does not optimize the CPU frequency and average transmit power of mobile users, and it also fails to consider the energy consumption associated with training or transmitting HFL models. Therefore, TSIA proves advantageous by considering these factors and achieving superior results in terms of minimizing time delay and energy consumption.

We now provide an elaboration on the convergence behavior of TSIA, as depicted in Fig.\ref{TSIA-iter}. As discussed in Section\ref{edge_allo}, TSIA achieves convergence when it revisits a previously encountered user assignment pattern. For instance, in Fig.~\ref{TSIA-iter}(a), the assignment pattern at the 13th assigning iteration matches that of the 11th iteration, indicating the convergence of TSIA. Furthermore, we investigate the relationship between the convergence rate of TSIA and the number of assigning iterations, as influenced by the quantities of mobile users ($N$) and edge servers ($M$). In Fig.\ref{change_N}, we fix $M$ at 5 and vary $N$ from 10 to 100. Similarly, in Fig.\ref{change_M}, $N$ is fixed at 50 while $M$ ranges from 2 to 10. The results reveal that TSIA requires a greater number of assigning iterations to achieve convergence when the count of mobile users or edge servers increases. Additionally, for $N = 50$ and $M = 5$, the number of assigning iterations typically falls within the range of 20 to 50. In comparison, JUARA and HFEL utilize 100 and 400 assigning iterations, respectively. This observation indicates that TSIA shows a considerably faster convergence speed compared with HFEL\footnote{For TSIA, HFEL, and JUARA, an assigning iteration corresponds to a single execution of the spectrum resource management methods. Hence, the convergence speed of TSIA, HFEL, and JUARA can be evaluated based on the number of assigning iterations.}.

\subsection{Comparison between HFL and traditional FL}\label{user_result}
The proposed HFL framework is compared with traditional FL. Both HFL and traditional FL are trained using the same mobile users, allowing for a direct comparison of the learning performance. The numbers of local update $L$ and edge aggregation $K$ in HFL are both 5. The number of local update in FL $L_\text{FL}$ is 5. The experiments are repeated five times. Fig.~\ref{datasets} presents the learning accuracy of the testing set throughout the training process. The accuracy curves of HFL and FL exhibit proximity, indicating that the learning performance of HFL is comparable to that of traditional FL. This finding suggests that transferring a portion of the model aggregation to the edge does not result in a deterioration of learning performance.

Additionally, we evaluate the objective value~\eqref{YY} $E_\text{sum}+\lambda T_\text{sum}$ of both HFL and FL, with $\lambda = 1$. The experiments are conducted five times, with variations in the locations of the mobile users and edge servers, as well as the values of $D_n$ and $c_n$. The cloud server remains fixed at the centre of the area. In FL, the mobile users solely communicate with the cloud server, and the bandwidth resource of this cloud server, denoted as $B_\text{FL}$, corresponds to the summation of the bandwidth resources of all the edge servers in HFL (i.e., $\sum_{m = 1}^{M}B_m$). Figure~\ref{compare_FL} presents the experimental results for both the FDMA and OFDMA schemes. It can be seen that the objective value~\eqref{YY} in HFL is always lower compared with traditional FL. This observation indicates that the communication and computation overheads are alleviated through the involvement of the edge servers. In summary, the proposed HFL framework enables the system to achieve lower system costs while maintaining learning performance compared with the traditional FL framework.

\section{Conclusions}\label{section7}
Hierarchical Federated Learning (HFL) has been proposed to address privacy concerns and relieve network congestion in FL. In this paper, we have studied the mechanism of HFL over wireless networks and investigated the user assignment problem in HFL. We formulated a user assignment problem to minimize the weights sum of the latency and energy consumption of HFL. To solve this problem, we proposed a two-stage iterative algorithm (TSIA) to handle the user assignment problem. TSIA iteratively transfers the mobile users from the straggler server to other servers, thus alleviating the communication and computation overheads of all the edge servers. As user assignment couples with the resource allocation problem, we proposed a spectrum resource optimization algorithm (SROA) to jointly optimize users' bandwidth, CPU frequency, and average transmit power within a single edge server. The experimental results showed that SROA outperforms other benchmarks in minimizing energy consumption and time delay of HFL. Besides, TSIA enables HFL to significantly reduce global costs with a faster convergence speed. Finally, compared with traditional FL, HFL achieves a lower global cost while ensuring learning performance.

\appendices
\section{Proof of LEMMA 1}
\begin{proof}
Let $h(x) = x\log_2(1+\frac{G_n}{x})$ ($x > 0$), the first and second order derivatives of $h(x)$ are derived as
\begin{align}
    h'(x) &= \log_2(1+\frac{G_n}{x}) - \frac{G_n}{\ln2 (x+G_n)}, \label{Qd}\\
    h''(x) &= \frac{-G_n^2}{x(x+G_n)^2\ln2} < 0. \label{Qdd}
\end{align}
$h'(x)$ is a decreasing function as $h''(x) < 0$. Besides, we can observe that $\text{lim}_{x \rightarrow +\infty}h'(x) = 0$, and thus $h'(x)$ is always larger than 0. Therefore, $h(x)$ is an increasing function. Since $\frac{1}{\ln2}(x-1) > \text{log}_2(x)$, we have $x\log_2(1+\frac{G_n}{x}) < x\frac{1}{\ln 2}\frac{G_n}{x} = \frac{G_n}{\ln2}$. According to~\eqref{Xa}, we have
\begin{align}
    f_n &\geq \frac{J_n}{t- \delta_n - \frac{H_n}{b_n\log_2(1+\frac{G_n}{b_n})}}, \label{Fd}\\
        &\geq \frac{J_n}{t- \delta_n  - \frac{\ln2 H_n}{G_n}}. \label{Fdd}
\end{align}
Therefore, the lower bound of $f_n$ can be derived as $\max(0, \frac{J_n}{t- \delta_n-\ln 2\cdot H_n/G_n})$.
\end{proof}

\section{Proof of LEMMA 2}
\begin{proof}
According to~\eqref{Ya}, we have
\begin{align}
    p_n &\geq \frac{2^{\frac{Y_n}{t - F_n}}-1}{Z_n}, \label{Fd}
\end{align}
To observe the relationship between $p_n$ and $b_n$, $Y_n = \frac{IKs}{b_n}$ and $Z_n = \frac{h_n}{N_0 b_n}$ are plugged into \eqref{Fd} as follows
\begin{align}
    p_n &\geq \frac{N_0\Big( b_n 2^{\frac{IK s}{b_n(t - F_n)}}-b_n \Big)}{h_n}, \label{FFd}
\end{align}
Let $y(x) = x 2^{\frac{\xi}{x}}-x$ ($x > 0$), where $\xi = \frac{IKs}{t - F_n}$. The first and second order derivatives of $y(x)$ are derived as
\begin{align}
    y'(x) &= 2^\frac{\xi}{x} + x 2^\frac{\xi}{x}\ln2\frac{-\xi}{x^2} - 1   \notag \\
    & = 2^\frac{\xi}{x} - \frac{\xi \ln2}{x}2^\frac{\xi}{x} -1
\end{align}
\begin{align}
    y''(x) &= 2^\frac{\xi}{x}\ln2\frac{-\xi}{x^2} - \xi\ln2(-\frac{1}{x^2}2^\frac{\xi}{x} + \frac{1}{x}2^\frac{\xi}{x}\frac{-\xi}{x^2}\ln2)\\
    &= \frac{(\xi\ln2)^2}{x^3}2^\frac{\xi}{x} > 0
\end{align}
Therefore, $y'(x)$ is an increasing function. If $x \rightarrow +\infty$, we have
\begin{align}
    &\lim_{x \rightarrow +\infty}y'(x)\\
    =&\lim_{x \rightarrow +\infty}(2^\frac{\xi}{x} - \frac{\xi\ln2}{x}2^\frac{\xi}{x} - 1)\\
    =& 1 - 0 - 1 = 0
\end{align}
Thus, $y'(x) < 0$ for $x > 0$, which indicates that $y(x)$ is a decreasing function. Replacing $x$ with $b_n$, the lower bound of $\frac{N_0}{h_n}y(b_n)$ is obtained when $b_n = b_\text{max}$. 

In terms of $p_n$ and $f_n$, \eqref{Fd} is rewritten as follows.
\begin{align}
    p_n &\geq \frac{2^{\frac{Y_n}{t - \delta_n - \frac{J_n}{f_n}}}-1}{Z_n}  \label{pd}
\end{align}
When $f_n > 0$, $t - \delta_n - \frac{J_n}{f_n}$ is an increasing function w.r.t. $f_n$. Therefore, the lefthand side of~\eqref{pd} is a decreasing function w.r.t. $f_n$. In this case, the lower bound of $p_n$ is obtained when $f_n = f^{\text{max}}_n$.

In conclusion, the lower bound of $p_n$ is achieved when $b_n = b_\text{max}$ and $f_n = f^{\text{max}}_n$.
\end{proof}

% \section{Proof of THEOREM 1}
% We take the first stage of TSIA as an example. Given $N$ mobile users and $M$ edge servers, TSIA goes through at most $M^N$ different assignment patterns $\Psi_1,...,\Psi_{M^N}$ based on Remark~\ref{totalpattern}. If TSIA continues assigning the mobile users after $M^N$ assigning steps, there must exist one unique assignment pattern $\Psi_q (q\in \mathbb{Z}^+$ and $q\in[1,M^N])$ that satisfies $\Psi_{M^N+1} = \Psi_q$, which indicates the assigning pattern returns to the $q$-th assigning step.

% As TSIA is a deterministic policy according to Remark~\ref{deterministic}, the assigning trajectory $\Psi_q,...,\Psi_{M^N}$ is always fixed. Therefore, the first stage of TSIA requires $M^N$ assigning steps to reach convergence in the worst case. Using the above analysis, it can be concluded that the second stage of TSIA also needs $M^N$ assigning steps to reach convergence in the worst case. That completes the proof.

\begin{spacing}{1}
 \renewcommand{\refname}{~\\[-20pt]References\vspace{-5pt} }
\bibliographystyle{IEEEtran}
\bibliography{sample-base}

\end{spacing}

\end{document}